\documentclass{article}

\usepackage[left=1.75cm, top=2cm, bottom=1.5cm,right=1.75cm]{geometry}

\usepackage{amsmath, amsfonts}
\usepackage[dvips]{graphicx}
\usepackage{xcolor}
\usepackage{theorem}
\usepackage[normalem]{ulem}
\usepackage[mathscr]{eucal}
\usepackage{url}

\newenvironment{proof}[1][Proof]{\textbf{#1.} }{\ \rule{0.5em}{0.5em}}

\newtheorem{thm}{Theorem}[section]

\newtheorem{cor}[thm]{Corollary}
\newtheorem{lem}[thm]{Lemma}
\newtheorem{prop}[thm]{Proposition}

\newtheorem{Assumption}[thm]{Assumption}

{\theorembodyfont{\rmfamily}\newtheorem{rem}{Remark}[section]}

\newcommand{\be}{\begin{equation}}
\newcommand{\mb}{\mathbf}
\newcommand{\p}{\partial}
\newcommand{\ee}{\end{equation}}
\newcommand{\bq}{\begin{eqnarray}}
\newcommand{\eq}{\end{eqnarray}}
\newcommand{\ex}{\mathbb{E}}
\newcommand{\half}{\frac{1}{2}}
\newcommand{\bs}{\bigskip}
\newcommand{\nind}{\noindent}
\newcommand{\nn}{\nonumber}
\newcommand{\pdd}[2]{\frac{\partial #1}{\partial #2}}
\newcommand{\pr}{\mathbb{P}}
\newcommand{\lb}{\lbrace}
\newcommand{\rb}{\rbrace}
\newcommand{\e}{\varepsilon}

\newcommand{\sk}{\smallskip}

\newcommand{\al}{\alpha}
\newcommand{\gm}{\gamma}
\newcommand{\lee}{\,\,\,\,\, \le \,\,\,\,\,}
\newcommand{\lt}{\,\,\,\,\, < \,\,\,\,\,}

\newcommand{\gt}{\,\,\,\,\, > \,\,\,\,\,}

\newcommand{\ul}{\underline}

\newcommand{\mc}{\mathcal}
\newcommand{~}{\,\,\,\,=\,\,\,\,}
\newcommand{\dl}{\delta}

\newcommand{\R}{\mathbb{R}}

\newcommand\xb{\bar{x}}
\newcommand\yb{\bar{y}}

\def \R {\mathbb{R}}
\def \p {\partial}

\newcommand\Eb{\mathbb{E}}

\newcommand\Rb{\mathbb{R}}

\newcommand\Ac{\mathscr{A}}

\newcommand\sig{\sigma}

\renewcommand\d{\partial}

\newcommand\lm{\lambda}

\newcommand{\<}{\langle}
\renewcommand{\>}{\rangle}

\newcommand{\bl}[1]{\textcolor{black}{#1}}

\newcommand{\ja}[1]{\textcolor{black}{#1}}

\newcommand{\mf}[1]{\textcolor{black}{#1}}

\begin{document}

\title{\textbf{Small-time asymptotics under local-stochastic \\
volatility with a jump-to-default: curvature \\and the heat kernel expansion}}
\bs

\author{ John Armstrong\thanks{Dept. Mathematics, King's College London, Strand, London, WC2R 2LS ({\tt John.1.Armstrong@kcl.ac.uk})}
 \and
Martin Forde\thanks{Dept. Mathematics, King's College London, Strand, London WC2R 2LS  ({\tt Martin.Forde@kcl.ac.uk})} \\
\and
Matthew Lorig\thanks{Dept. of Applied Mathematics,
University of Washington, Seattle, WA, USA
({\tt mlorig@uw.edu})
 }
\and
Hongzhong Zhang\thanks{Dept. IEOR,
Columbia University,
New York, NY 10027 ({\tt hz2244@columbia.edu})}
}

\date{\today }

\maketitle

\begin{abstract}
\nind We compute a sharp small-time estimate for implied volatility under a general uncorrelated local-stochastic volatility model.  For this we use the Bellaiche \cite{Bel81} heat kernel expansion combined with Laplace's method to integrate over the volatility variable on a compact set, and (after a gauge transformation) we use the Davies  \cite{Dav88} upper bound for the heat kernel on a manifold with bounded Ricci curvature to deal with the tail integrals.  If the correlation $\rho < 0$, our approach still works if the drift of the volatility takes a specific functional form and there is no local volatility component, and our results include the SABR model for $\beta=1, \rho \le 0$.  \bl{For uncorrelated stochastic volatility models, our results also include a SABR-type model with $\beta=1$ and an affine mean-reverting drift, and the exponential Ornstein-Uhlenbeck model.} We later augment the model with a single jump-to-default with intensity $\lm$, which produces qualitatively different behaviour for the short-maturity smile;  in particular, for $\rho=0$, log-moneyness $x > 0$,  the implied volatility increases by $\lm f(x) t +o(t) $ for some function $f(x)$ which blows up as $x \searrow 0$.  Finally, we compare our result with the general asymptotic expansion in Lorig, Pagliarani  \&  Pascucci  \cite{LPP15}, and we verify our results numerically for the SABR model using Monte Carlo simulation and the exact closed-form solution given in Antonov  \&  Spector  \cite{AS12} for the case $\rho=0$.
\footnote{The authors would like to thank Elton Hsu, Alan Lewis, Robert Neel, Cheng Ouyang and Louis Paulot for many valuable insights, and a particular thanks to the late Peter Laurence for drawing MF's attention to the Bellaiche heat kernel estimate.  The authors are also grateful to three referees and one associate editor, whose comments and suggestions have improved this manuscript.}
\end{abstract}


\section{Introduction}

In the physics literature, a very convenient form for the heat kernel was originally given by de Witt \cite{dW65} (see also  McAvity  \&  Osborn  \cite{MO91}).  We can re-write the second order elliptic operator associated with a general diffusion process on $\mathbb{R}^n$ in terms of the Laplace-Beltrami operator plus a first order differential operator (i.e. a vector field); the heat kernel expansion is obtained as the exponential of the work done by the vector field along the geodesic joining the \bl{start and the end} points, multiplied by the Minakshisundaram-Pleijel \cite{MP49} heat kernel expansion for the usual Laplace-Beltrami operator, which contains the leading order exponential term from large deviations theory multiplied by the square root of the Riemannian volume form element under geodesic spherical coordinates (see also Chavel \cite{Chav84}, Hsu \cite{Hsu02}, Laurence \cite{Laur10}, Neel \cite{Neel07}, Molchanov \cite{Mol75}).   \cite{Mol75} provides a rigorous probabilistic proof of the de Witt expansion at leading order, using a Girsanov change of measure and conditioning on the end point of the process i.e. considering the bridge process.  Bellaiche \cite{Bel81} showed that the Molchanov expansion also holds for non-compact manifolds under certain technical conditions, and it is
the latter which we use in this article.

\sk \sk

Paulot  \cite{Pau10} formally derived a small-time expansion for call options under a general local-stochastic volatility model by applying Laplace's method to integrate the heat kernel over the range of the volatility variable, then using the Tanaka-Meyer formula and some well known asymptotic expansions for the standard \ja{n}ormal distribution function.   \cite{Pau10} also computes explicit formulae for the well known SABR model.  Henry-Labord\'{e}re  \cite{HL08} formally derived a small-time expansion for the implied volatility from a small-time expansion for the \textit{effective local volatility}.  However, both authors do not justify integrating over the infinite range of the volatility variable with appropriate tail estimates (which is needed because the aforementioned small-time heat kernel expansions are only known to converge uniformly on compact sets).  This is the main technical issue which we resolve in this article, under suitable conditions on the vol-of-vol coefficient $\al(y)$ and drift coefficient $\mu(y)$ as $y \to 0$ and as $y \to \infty$.


\sk

 \cite{FJ11} characterize the small-time behaviour of the implied volatility (at leading order) for a local-stochastic volatility model with zero correlation, using the Freidlin-Wentzell
theory of large deviations for SDEs and then converting to the differential geometry problem of computing the
shortest geodesic from a point to a vertical line on a Riemannian manifold.  The volatility is assumed to be bounded, which means that the curvature can change sign (unlike the SABR model), and the solution to this variable endpoint problem is obtained using conserved quantities which arise from integrating the geodesic equations and a transversality condition, where the shortest geodesic comes in perpendicular to the vertical line under the aforementioned metric.  The small-time behaviour of the price of an out-of-the-money European call option is computed using H\"{o}lder's inequality, and this is then translated into a statement about the small-time behaviour of the implied volatility.   \cite{FJ11} also derive a series expansion for the implied volatility in the small-maturity
limit, in powers of the log-moneyness, and they show how to calibrate such a model to the observed implied volatility
smile in the small-maturity limit.

\sk \sk
 Gatheral et al. \cite{GHLOW12} consider small-time asymptotics for a one-dimensional local volatility model, using Girsanov's theorem and conditioning with a bridge process to derive a small-time expansion for the transition density which holds uniformly in $\mathbb{R}$.  They also derive the corresponding expansion for the implied volatility.  When the diffusion coefficient is time-dependent, they find that even the leading order term in the expansion requires a small but important modification.  For a time-homogenous one-dimension diffusion process $dS_t=S_t \sigma(S_t) dW_t$, they prove the following asymptotic expansion for the implied volatility $\hat{\sigma}(K,t)$ at strike $K$ and time-to-maturity $t$
\bq
\hat{\sigma}(K,t) &=&\hat{\sigma}_0(K)\, +\,\frac{\hat{\sigma}_0(K)^3}{\left(\log \frac{K}{S_0}\right)^2}\log\frac{\sqrt{\sigma(S_0)\sigma(K)}}{\hat{\sigma}_0(K)}t\,+\,O(t^2)\, \nn
\eq
\nind as $t \to 0$, where $\hat{\sigma}_0(K)=\Big(\frac{1}{\log \frac{K}{S_0}}\int_{S_0}^{K}\frac{du}{u\sigma(u)}\Big)^{-1}$ is the well known leading order term (see also Busca et al.  \cite{BBF},  \cite{BBF2}).

\sk \sk

Deuschel et al. \cite{DFJV11}, \cite{DFJV11b} use Laplace's method on Wiener space in the same spirit as Azencott, Bismut and Ben Arous \cite{BA88} to compute a small-noise expansion for the density of the canonical projection into $\mathbb{R}^m$ of an $n$-dimensional hypoelliptic diffusion process $X^{\e}_t$.  This is accomplished using the Ben Arous expansion applied to the characteristic function of $X^{\e}_t$ combined with a Fourier inversion.  This has the advantage over the heat kernel expansion that the diffusion coefficient need not be uniformly elliptic.  However, they do not compute the pre-factor that goes in front of the expansion, which is needed to compute the correction term for implied volatility in the small-time limit, which is computed in this article.  More recently, Friz \& deMarco \cite{FdM13} consider a stochastic volatility model governed by a hypoelliptic diffusion satisfying a strong Hormander condition, and in this setting they compute a Varadhan-type formula for the small-time behaviour of the stock price density and characterize the small-time behaviour of the law of the terminal stock price conditioned on the terminal volatility, from which we can then compute the effective local volatility $\mathbb{E}(\sigma_t^2|S_t=K)$ which is a fundamental quantity in the mimicking theorems of Gy\"{o}ngy \cite{Gyo86} and the more recent work of Brunick \& Shreve \cite{BS13}.

\sk \sk

Lorig et al. \cite{LPP15} derive a full asymptotic expansion for the price of a call option and the associated implied volatility under a general class of local-stochastic
volatility models, and provide a rigorous error bound under a uniform ellipticity condition on the diffusion coefficient for the model.  Their error bound is obtained using Duhamel's principle and classical estimates by Friedman on partial derivatives of the fundamental solution to the inhomogeneous heat equation $[\p_t + a(x) \p^2_x] u = 0$ in terms of the fundamental solution to the standard (homogenous) heat equation $[\p_t + \half \sigma^2 \p^2_x] u = 0$ for $\sigma$ constant.  The uniform ellipticity condition is relaxed in Pagliarani \& Pascucci \cite{PP14}, who consider a general class of degenerate (i.e. non-uniformly parabolic) PDEs, which includes the CEV, Heston and SABR models, and
hybrid credit-equity models such as the JDCEV model, and they again derive a rigorous error bound for small-times.

\sk
\sk

\subsection{Outline of article}

\bs
 In Theorem \ref{thm:Bellaiche}, we recall the Bellaiche \cite{Bel81} small-time heat kernel expansion.  This is the main result that is used for proving the main result (Theorem \ref{thm:SmallTCallsOTM}) where we compute a small-time expansion for non at-the-money call options under a general local-stochastic volatility model with zero correlation. The pre-factor in the heat kernel expansion is expressed in terms of the Jacobian of the exponential map, which is a ratio of the pullback of two volume forms at the start and the end point of the geodesic.  We show how this calculation is simplified by working in geodesic normal coordinates, and explain the geometric meaning of curvature as the first order deviation from the Euclidean metric in geodesic normal coordinates.

 \bs
   In section \ref{section:TheModel} we introduce our local-stochastic volatility model, and we
compute the Laplace-Beltrami operator, the metric and the curvature associated with the model (the metric is induced by the inverse of the diffusion coefficient matrix).  We then state the technical assumptions on the coefficients in the SDE\ja{s}, the most important of which is that $\al(y) \sim A_1 y$ as $y \to 0$ and $\alpha(y) \sim B_1 y^p$ as $y \to \infty$ for some constants $A_1,B_1>0$ and $p \in (0,1]$, where $\al(y)$ is the vol-of-vol coefficient.  Here $\sim$ means that the ratio between the two sides tends to 1 in the limit. This ensures that the associated Riemannian manifold on the upper half plane is complete\ja{: }the distance to $y=0$ and $y=\infty$ under the metric $g_{ij}$ is infinite, \ja{which} ensures that $y=0$ and $y=\infty$ are unattainable boundaries for the volatility process $Y_t$.  We then impose that the manifold $M$ has negative curvature (recall that in two dimensions the sectional and the Gaussian curvature are the same, which by Hadamard's theorem, implies that the cut locus of $M$ is empty), and we discuss some simple well known examples of parametric stochastic volatility models.

 \bs

    In Theorem \ref{thm:SmallTCallsOTM}, we give the main result of the article - a small-time expansion for non at-the-money call options under the aforementioned model.  This effectively sharpens the result obtained in  \cite{FJ11} and relaxes the assumption of bounded volatility to allow for more realistic tail behaviour (e.g.\ moment explosions).  \bl{Our model setup can include an extended SABR-type model with $\rho=0$ and an affine mean-reverting drift coefficient and $\beta$ parameter equal to 1, and the exponential Ornstein-Uhlenbeck model (see Table \ref{table:sv})}.  The proof follows the steps of Paulot \cite{Pau10} but with full rigour\ja{. W}e use Laplace's method to integrate the two-dimensional heat kernel with respect to Lebesgue measure multiplied by the local volatility squared to compute a small-time expansion for $\frac{1}{dx}\mathbb{E}(\sigma(X_t)^2 Y_t^2 1_{X_t \in dx})$ where $Y_t$ is the volatility and $X_t$ is the log forward price.  The tail integrals are dealt with using the Davies \cite{Dav88} upper bound for the heat kernel on a Riemannian manifold with Ricci curvature bounded from below combined with a gauge transformation.  We then use the Tanaka-Meyer formula to estimate the price of a call option in terms of $\frac{1}{dx}\mathbb{E}(\sigma(X_t)^2 Y_t^2 1_{X_t \in dx})$ by integrating over time from zero to the maturity of the option, and using well known asymptotic results for the standard \ja{n}ormal distribution function.  This trick using the gauge transformation only works when $\rho=0$, unless we impose a specific functional form for the drift of the volatility process and we assume that the local volatility function $\sigma(x)$ is constant (see subsection \ref{subsection:RhoNotZero} for details on this)\ja{. }In particular we show that SABR model with $\beta=1,\rho \le 0$ can still be handled using this trick.  To the best of our knowledge, this is the first rigorous analysis of the small-time correction term for implied volatility under SABR model (the leading order term is computed using viscosity solutions in  \cite{BBF2}).

    \bs
    In Appendix \ref{section:CalculatingA}, we discuss how to explicitly compute the drift correction term that appears in the Bellaiche expansion, and when there is no local volatility component, this term takes an especially simple form.  In principle, we can also formally derive a similar expansion for at-the-money call options, for which the small-time behaviour is qualitatively different, but this requires knowledge of the next term in Bellaiche heat kernel expansion (which is not given by Bellaiche), and we do not have a published reference for this next term, aside from slides by Laurence \cite{Laur08}, so we defer the details for future work.

 \bs
In Proposition \ref{propn:BSSmallT}, we derive the corresponding expansions for call options using the Black-Scholes call option formula but with a time-dependent volatility function.  This is needed in section 5, where we derive the correction terms for the implied volatility of non at-the-money options.  The correction term for implied volatility is important because it takes account of the drift terms in the SDEs, which the result in  \cite{FJ11} fails to capture because Freidlin-Wentzell theory only works on a crude logarithmic scale.  The correction term is also required to accurately approximate the price of a call option at small maturities.  In section 7 we give closed-form formulae for all expressions of interest for the well known SABR model, and we verify our implied volatility expansion numerically by comparing against Monte Carlo simulations and the closed-form expression for the price of a call option under the SABR model with $\rho=0,\beta=1$ given in  \cite{AS12}\ja{. W}e find they are all in very close agreement.  Finally, in section 8, we enrich the model by adding a single Poisson jump-to-default with hazard rate $\lm$, and when $\rho=0, \sigma(x) \equiv 1$, we show that the jump-to-default increases the implied volatility at log-moneyness $x$ by the following amount
\bq
\label{eq:cor_}
 \lm \frac{\hat{\sigma}^2(x)}{|x| \, y_1^*(x)} t +o(t)
\eq
for $x>0$ as $t\to 0$, where $\hat{\sigma}(x)$ is the zeroth order implied volatility and $y_1^*(x_1)$ is the $y$-value of the intersection point for the shortest geodesic from the point $(x_0,y_0)$ to the vertical line $\lb x=x_1 \rb $ under the metric induced by the diffusion coefficient for the model, and the correction term in \eqref{eq:cor_} blows up as the log-moneyness $x \searrow 0$.

\bs

Throughout the article, we write $a=o(b)$ if and only if the limit of $a/b$ is 0 as some parameter of interest tends to zero or infinity (depending on the context), and write $a=O(b)$ if and only if $\limsup |a/b|<\infty$. If $a=O(b)$ and $b=O(a)$ hold simultaneously, we write $a\asymp b$.

\bs

\section{The Heat kernel expansion}
\label{section:LBandHeatKernel}

\sk
\sk

\sk
Consider a diffusion process on $\mathbb{R}^n$ with infinitesimal generator $L$.  In local coordinates, $L$ takes the form
\bq
\label{eq:SK}
L ~\frac{1}{2}\sum_{1 \le i,j\le n}a^{ij}(\mb{x})\pdd{^2}{ x_i \partial x_j}+ \sum_{1 \le i \le n} b^i(\mb{x}) \pdd{}{x_i}\,\quad\mb{x}=(x_1,\ldots,x_n)\in\mathbb{R}^n.\nn
\eq
\nind Let $M=\mathbb{R}^n$ with metric $(g_{ij})=(a^{ij})^{-1}$ so that $M$ is a smooth Riemannian manifold with a single chart given by the identity map.  We can write $L$ as $\frac{1}{2}\Delta+\mathcal{A}$, where $\Delta=\sum_{i,j}\frac{1}{\sqrt{|g|}}\partial_i (\sqrt{|g|}\,g^{ij}\partial_j)$ is the Laplace-Beltrami operator and $\mathcal{A}^i=b^i-\frac{1}{2}\sum_j \frac{1}{\sqrt{|g|}}\partial_j (\sqrt{|g|}\,g^{ij})$ is a smooth first order differential operator and $|g|=|\det g_{ij}|$ (recall that $(g^{ij})=(g_{ij})^{-1}$).

\bs

    The heat kernel of $L$ is a continuous function $p_t(\mb{x},\mb{y})=p(\mb{x},\mb{y},t)$ defined on $M\times M \times (0,\infty)$ which is twice differentiable in $\mb{x}$ and once differentiable in $t$, which satisfies the backward Kolmogorov equation
  \bq
  (-\d_t + L )p ~ 0 \,,
  \eq
  \nind 
\footnote{Technically this is a forward equation in the $t$ variable and a backward equation in $\textbf{x}$, but it becomes a forward equation if we make the usual $t \mapsto T-t$ transformation.} such that for any bounded continuous function $f$ with compact support we have
\bq
\lim_{t \to 0}   \int_M \hat{p}_t(\mb{x},\mb{y})f(\mb{y})   d\mb{y} &=& f(\mb{x})\,,\quad \mb{y}=(y_1,\ldots,y_n)\in M,
\label{eq:Dirac}
\eq
where $\hat{p}_t(\mb{x},\mb{y}):=p_t(\mb{x},\mb{y})\sqrt{|g(\mb{y})|}$
\footnote{Note that we typically write \eqref{eq:Dirac} informally as $\hat{p}_0(\mb{x},\mb{y})=\delta(\mb{y}-\mb{x})$, but \eqref{eq:Dirac} really says that $\hat{p}_t(\mb{x},\mb{y})d\mb{y}$ tends to the Dirac measure $\delta_{\mb{x}}(\mb{y})$ in the sense of weak convergence as $t \to 0$.}(see e.g. page 135 in Chavel\cite{Chav84} for details, or Definition 5.7.9 in \cite{KS91} for a non-geometric reference).  $\hat{p}_t(\mb{x},\mb{y})$ is the probability density of $X_t$ conditioned on $X_0=\mb{x}$, with respect to Lebesgue measure $d\mb{y}=dy_1\ldots dy_n$ (see e.g. equation (5.7.26) in \cite{KS91}).  Furthermore,  if $\mc{A}=0$, then $p_t(\mb{x},\mb{y})$ is symmetric in $\mb{x}$ and $\mb{y}$ (see e.g. Theorem 1, chapter VI on page 138 of \cite{Chav84}).

\sk
\begin{rem}
 For a simple example to show the connection between the heat kernel $p_t(\mb{x},\mb{y})$ and the transition density $\hat{p}_t(\mb{x},\mb{y})$, consider a one-dimensional diffusion which satisfies $dX_t=\half\sigma^2X_tdt+\sigma X_tdW_t$ and $X_0=x\in\mathbb{R}_+$, where $\sigma>0$ is a constant and $W$ is a standard Brownian motion. Then one can easily verify that the Laplace-Beltrami operator $\Delta=\sigma^2x\,\d_x+\sigma^2x^2\d^2_x $ and $\mc{A}=0$.
Moreover, the Riemmanian metric $g_{ij}$ on $\mathbb{R}$ associated with this model has a single element $g_{11}(x)=\frac{1}{\sigma^2x^2}$, so the pre-factor $\sqrt{|g(\mb{y})|}$ is  $\frac{1}{\sigma y}$. On the other hand, we know that the transition density for $X$ is given by $\hat{p}_t(x,y)=\frac{1}{\sigma y\sqrt{2\pi t}}e^{-(\log\frac{y}{x})^2/(2\sigma^2 t)}$. Thus, the heat kernel is $p_t(x,y)=\frac{1}{\sqrt{2\pi t}}e^{-(\log\frac{y}{x})^2/(2\sigma^2 t)}$, a function that is symmetric in $x$ and $y$
 (see e.g. Chapters 5 and 6 of \cite{HL08} for more examples of this nature).\end{rem}
\sk


Intuitively, one expects that the heat kernel on a Riemannian manifold should be a deformation of the heat kernel on Euclidean space. Molchanov  \cite{Mol75} made this idea precise by providing a small-time expansion for the heat kernel on a compact manifold.  Subsequent authors have extended this result to more general manifolds. We state two such results.

\sk

\begin{thm} \label{thm:Hsu}(Theorem 5.1.1 in  \cite{Hsu02}).
Let $M$ be a complete $n$-dimensional Riemannian manifold. Let $C(M)$ be the subset of points $(\mb{x},\mb{y})$ in $M \times M$ such that $\mb{x}$ lies in the cut locus of $\mb{y}$, which we denote by $\mathrm{Cut}(\mb{y})$.  Let $d(\mb{x},\mb{y})$ be the Riemannian distance between two points $(\mb{x},\mb{y}) \in (M \times M ) \setminus C(M)$. Let $p_t(\mb{x},\mb{y})$ denote the heat kernel of $\frac{1}{2} \Delta$ on $M$ (i.e. $\mc{A}=0$).  Then there exist smooth functions $u_n(\mb{x},\mb{y})$ defined on $(M \times M) \setminus C(M)$
such that the following asymptotic expansion
\bq
 p_t(\mb{x},\mb{y}) \,\,\,\, \sim \,\,\,\, (2 \pi t)^{-n/2} e^{-d(\mb{x},\mb{y})^2/(2t)} \sum_{i=0}^\infty u_i(\mb{x},\mb{y}) t^i
 \eq
holds uniformly as $t \to 0$ over any compact subset of $(M \times M) \setminus C(M)$.
Let $\exp_{\mb{x}}:T_{\mb{x}} M \longrightarrow M$ be the exponential map based at $\mb{x}$, then we have
\bq  u_0(\mb{x},\mb{y}) ~ \left(J( \exp_{\mb{x}} )(Y)\right)^{-1/2}, \quad  \quad Y ~ \exp_{\mb{x}}^{-1} \mb{y}
\label{eq:u0}
\,.\eq
Here $J( \exp_{\mb{x}})$ denotes the Jacobian of the exponential map where we use the
flat metric induced by $g$ on $T_{\mb{x}} M$ to define the Jacobian (see Remark \ref{rem:2.1} below).
\end{thm}

\sk
\begin{rem}
\label{rem:2.1}
Recall that if $f:M_1 \longrightarrow M_2$ is a differentiable map between oriented Riemannian $n$-manifolds $(M_1,g_1)$ and $(M_2,g_2)$, then $J(f)(p)$ is defined to be the ratio of the pullback of the volume form on $M_2$ to the volume form on $M_1$ at $p$. Let us explain this definition in detail. The differential of $f$ at $p$ defines a map $f_*:T_p M_1 \longrightarrow T_q M_2$ where $q=f(p)$. The pullback $f^*: \Lambda^k T^*_q M_2 \longrightarrow \Lambda^k T^*_p M_1$ is then defined by:
\bq (f^*\mu)(v_1 \wedge \ldots\wedge v_k) &=& \mu( f_* v_1, \ldots, f_* v_k), \eq
where $v_1, \ldots v_k \in T_p M_1$ and $\mu$ is a $k$-form. Now take local coordinates ${x}$ for $M_1$ and ${y}$ for $M_2$ centered on $p$ and $q$. The space $\Lambda^n T^*_p$ is one dimensional and hence spanned by the volume form $\sqrt{|g_1|}\, dx^1 \wedge \ldots \wedge dx^n$. Thus for some $\lambda \in \R$ we have:
\bq
\label{eq:5}
f^*( \sqrt{|g_2|}\, dy^1 \wedge \ldots \wedge dy^n) &=& \lambda \sqrt{|g_1|}\, dx^1 \wedge \ldots \wedge dx^n \,.
\eq
By definition $J(f)(p)$ is equal to $\lambda$. See subsections \ref{subsection:NormalCoords} and 2.2 to see how the expression for $J( \exp_{\mb{x}} )(Y)$ simplifies when we work in geodesic normal coordinates.
\end{rem}

\bs

We now recall the following extension of Theorem \ref{thm:Hsu} by Bellaiche \cite{Bel81}:

\sk
\sk

\begin{thm} \label{thm:Bellaiche} (Theorem 4.1 in  \cite{Bel81}).  Let $M$ be a $C^4$-Riemannian manifold and $\mathcal{A}$ a $C^4$-vector field. Then the heat kernel $p_t(\mb{x},\mb{y})$ of the operator $\frac{1}{2}\Delta + \mathcal{A}$ satisfies:
\bq \label{eq:MolchanovFormula}
 p_t(\mb{x},\mb{y}) \,\,\,\, \sim \,\,\,\,  (2\pi t)^{-n/2} u_0(\mb{x},\mb{y}) \,e^{-\half d(\mb{x},\mb{y})^2/t+A(\mb{x},\mb{y})} \quad \quad \quad \quad (t \to 0)
\eq
for $(\mb{x},\mb{y}) \in (M \times M) \setminus C(M)$. Here $u_0$ is defined as in \eqref{eq:u0} and
\bq
A(\mb{x},\mb{y})\,\,\,\,:= \,\,\,\,\int_0^1 \langle \mathcal{A},\dot{\gamma}(s) \rangle \,ds
\label{eq:Aofxandy}
\eq
   and $\gm$ is the unique distance-minimizing geodesic \footnote{All geodesics have constant speed, and here we are just choosing the speed so the geodesic reaches point $\textbf{y}$ in unit time}  $\gamma:[0,1]\longrightarrow M$ joining $\mb{x}$ and $\mb{y}$. The estimate \eqref{eq:MolchanovFormula} is uniform on any compact subset of $(M \times M) \setminus C(M)$.
\end{thm}

\sk

\subsection{Geodesic normal coordinates and the geometric meaning of curvature}
\label{subsection:NormalCoords}
\bs

Now let $e_1,..,e_n$ be a basis of $T_p M$ which is orthonormal with respect to the scalar product on $T_p M$ induced by the metric $g_{ij}$.  For each vector $v \in T_p M$, writing its components with respect to this basis, we obtain a map
$\phi : T_p M  \mapsto  \mathbb{R}^n$, $v=v^i e_i \mapsto (v^1,...,v^n)$.  Then one has the associated geodesic normal coordinate system ${y}$ given by ${y}=\phi\circ \exp_{\mb{x}}^{-1}$ (see page 21 in  \cite{Jost09}). In these coordinates, all first order partial derivatives of the metric vanish at zero, i.e. $g_{ij,k}(0)=0$ for all $i,j,k$, and $\Gamma^i_{jk}=0$, and the metric has the following Taylor expansion:
\bq
 g_{ij} &=& \delta_{ij} \,- \,\sum_{a,b} \frac{1}{3}( R_{i a j b} + R_{i b j a} )y_a y_b + O(|{y}|^3)  \nn
 \eq
(see e.g. \cite{MA73}), where $R^{i}_{jkl}$ is the Riemann curvature tensor and we are following standard conventions for raising and lowering indices using the metric. This formula provides the basic geometric interpretation of curvature as the deviation of the metric from the Euclidean metric in normal coordinates.  Indeed, Riemann originally introduced  curvature using such an expansion; the definition using the Levi-Civita connection was only introduced later.

\bs

\subsection{Calculating $J( \exp_{\mb{x}} )(Y)$ and $u_0(\mb{x},\mb{y})$ in normal coordinates}
\label{subsection:JinNormalCoords}

\bs

If we now consider Remark \ref{rem:2.1} for special case when $M_1=T_{\textbf{x}} M$, $M_2=M$, $f=\exp_{\textbf{x}}$, $p=\textbf{x}$ and $q=\mb{y}=\exp_{\textbf{x}}(Y)$, and take coordinates for the tangent space ${\phi}: M_1= T_{\textbf{x}} M \to \mathbb{R}^n$ as above and normal coordinates ${y}:M \to \mathbb{R}^n$ on $M$, then written in these coordinates $f$ and $f_*$ are just the identity function.  \ja{The metric $g_1=\delta_{ij}$. The metric $g_2=g_{ij}$ is the metric associated with using normal coordinates on $M$ (which is the identity at $\textbf{x}$). Equation \eqref{eq:5} becomes:}
\bq
\exp^* \left(\sqrt{|g|} \, dy^1 \wedge ... \wedge dy^n \right) &=& \lambda \, d\phi^1 \wedge ... \wedge d\phi^n \,.\label{eq:10_page6}
\eq
Using that $\exp^*$ is also the identity, we see that in these coordinates $\lm=J( \exp_{\mb{x}} )(Y)=\sqrt{|g|}$.

\bs
From this we also obtain
\begin{eqnarray*}
u_0(\mb{x},\mb{y}) & = & ( \sqrt{|g|} )^{-\frac{1}{2}} \\
         & = & \bigg[ 1 \,-\, \sum_{i,a,b} \frac{1}{3}( R_{i a i b} + R_{i b i a}) y_a y_b + o(|{y}|^2) \bigg]^{-\frac{1}{4}} \\
         & = & 1 + \frac{1}{12} \sum_{i,a,b} (R_{i a i b} + R_{i b i a})y_a y_b + o(|{y}|^2)\,.
\end{eqnarray*}
 We now specialize to the two-dimensional case. The symmetries of the curvature tensor tell us that it has only one independent component. At the origin in normal coordinates one has $R_{1212}=R_{2121}=-R_{1221}=-R_{2112}=\kappa$, where $\kappa$ is the Gaussian curvature of the surface. All the other components of $R$ vanish.  Thus we have
\begin{eqnarray*}
u_0(\mb{x},\mb{y}) & = & 1 + \frac{1}{12} \kappa |{y}|^2 + o(|{y}|^2) \nn \,.
\end{eqnarray*}
Alternatively, one can introduce geodesic polar coordinates $y=(r \cos( \theta), r \sin( \theta ))$ in which case one obtains the same expansion but with $|y|$ replaced by $r$ and now $r=d(\mb{x},\mb{y})$ is the Riemannian distance.

\bs

\subsection{Geodesic polar coordinates}

\bs

 Let $n=2$ and let $f(r,\theta)=\exp_p(r v(\theta))$ for $p \in M$, with $0< r < i(p)$ (where $i(p)=d(p,C_p)$ is the injectivity radius at $p$) and $-\pi <\theta  \le \pi$, where $|v(\theta)|=1$, $|v'(\theta)|=1$.  $(r,\theta)$ are called geodesic polar coordinates at $p$.
\nind In these coordinates, the metric $\mathbf{g}(r,\theta)$ has coefficients
$$
g_{rr}\,=\, \left|\pdd{f}{r} \right|^2\,=\,|v(\theta)|^2=1, \, g_{r\theta}\,=\,0\,\,,\,\,g_{\theta\theta}\,=\, \left|\pdd{f}{\theta} \right|^2\,
$$
\nind   (see page 122 in do Carmo \cite{doC92}), where $|\cdot|$ refers to the norm under the original metric $g$.  Then $J(r)=\pdd{f}{\theta}|(r,\theta)$ is a \textit{Jacobi field}, and $J$ has the asymptotic behaviour $|J|^2=r^2-\frac{1}{3}\kappa r^4+ o(r^4)$ as $r \to 0$.  Thus the metric $\mathbf{g}(r,\theta)$ can be locally approximated as
\bq
ds^2 &\approx& dr^2 + r^2 \left( 1-\frac{1}{3}\kappa r^2 \right) d\theta^2\, ,\nn
\eq
\nind for $r \ll 1$, and we see that the curvature describes the departure of the metric from the usual metric $ds^2=dr^2+r^2 d\theta^2$ for polar coordinates in $\mathbb{R}^2$.

\sk
\begin{rem}
The heat kernel can be constructed geometrically by the method of parametrix starting from an approximate heat kernel in local coordinates.  Page 148 in Chavel \cite{Chav84} gives a nice sketch proof using geodesic polar coordinates.  In these coordinates, we can write the $u_0(\mb{x},\mb{y})$ term in the heat kernel expansion as
\bq u_0(\mb{x},\mb{y})~\bigg(\frac{\sqrt{\mb{g}(r,\theta)}}{r}\bigg)^{-\frac{1}{2}} \nn \,.
\eq
\end{rem}

\begin{rem}
If $d(\mb{x},\mb{y})<d(\mb{x},\mathrm{Cut}(\mb{x})$), then in local coordinates $u_0(\mb{x},\mb{y})=\mf{\sqrt{\Delta^{\mathrm{VVM}}(\mb{x},\mb{y})}}$, where
\bq
\label{eq:VVM}
\Delta^{\mathrm{VVM}}(\mb{x},\mb{y})~g(\mb{x})^{-\frac{1}{2}}\det \left( -\pdd{^2 \phi(\mb{x},\mb{y})}{x_{i} \partial y_{j}} \right)g(\mb{y})^{-\frac{1}{2}}
\eq
\nind is the so-called \textit{van Vleck-Morette determinant} (see McAvity \& Osborn \cite{MO91} and equation (4.38) in Vassilevich \cite{Vass03}), and
$\phi(\mb{x},\mb{y})=\frac{1}{2}d(\mb{x},\mb{y})^2$.  \eqref{eq:VVM} is useful when we can compute $d(\mb{x},\mb{y})$ explicitly by solving the geodesic equations.
\end{rem}

\sk
\sk

\section{Local-stochastic volatility models}
\label{section:TheModel}
\bs

\label{subsection:TheModel}

We work \bl{with a probability space} $(\Omega,\mc{F},\mathbb{P})$ throughout, with a filtration $\mc{F}_t$ supporting two independent Brownian motions which satisfies the usual
conditions.

\sk
\sk

We now consider a general uncorrelated local-stochastic volatility model for a forward price process $S_t$ defined by
the following stochastic differential equations for $X_t=\log S_t$:
 \bq
\label{eq:Model}
        \left\{
        \begin{array}{ll}
        dX_{t}\,=\,-\frac{1}{2}\sigma(X_t)^2Y_t^2 dt+\sigma(X_t)Y_t \,
dW^1_t\,,\\
        dY_{t}\,\,=\,\mu(Y_t)dt+  \alpha(Y_t) dW^2_t,
        \end{array}\
        \right.
        \eq

\nind where $(X_0,Y_0)=(x_0,y_0)\in\mathbb{R}\times\mathbb{R}_+$, and $W^1, W^2$ are two independent standard Brownian motions.  We need to impose that the correlation is zero for the gauge transformation trick in subsection \ref{subsection:Gauge} to work.  However, the presence of the local volatility component $\sigma(x)$ can still produce an implied volatility skew.

\bs
 We let $M$ denote the upper half plane $\lb (x,y):y>0 \rb$ with Riemmanian metric $(g_{ij})=(a^{ij})^{-1}$, where $a_{ij}$ is the diffusion coefficient for the model in \eqref{eq:Model}, so the line element for $g$ is given by
\bq\label{eq:SVmetric} ds^2 ~  \sum_{i,j} g_{ij} dx_i dx_j ~  \frac{1}{\sigma(x)^2 y^2}dx^2\,+\, \frac{1}{\al(y)^2 }dy^2\,,\eq
and the Laplace-Beltrami operator $\Delta=\sum_j \frac{1}{\sqrt{|g|}}\partial_j (\sqrt{|g|}\,g^{ij})$ satisfies
\bq
\label{eq:LBoperator}
\half \Delta &=&
\half  y^2 \sigma'(x)\sigma(x)\p_x \,+\, \half \left[\al'(y) \al(y)  -\frac{\al(y)^2}{y} \right] \p_y \,+\, \half y^2 \sigma(x)^2 \p^2_x\,+\, \half \alpha(y)^2 \p^2_y \,,
\eq
so we have
\bq
\mc{A} &=&  \left[-\frac{1}{2}y^2 \sigma(x)^2-\half  y^2 \sigma'(x)\sigma(x) \right] \p_x\,+\, \left[\mu(y) \,-\, \half (\al'(y) \al(y)  -\frac{\al(y)^2}{y}) \right] \p_y \nn \,.
\eq

\begin{rem}
See Appendix \ref{section:CalculatingA} for details on how to calculate $A(\mb{x},\mb{y})$ \mf{as defined in \eqref{eq:Aofxandy}}.
\end{rem}

\sk

 We can easily compute the curvature tensor for the metric $g$ directly from the standard formulae for the Christoffel symbols in local coordinates and standard formulae for the curvature tensor $R$.\footnote{The  Mathematica sheet to calculate the curvature is available on request.}  We can then compute the Gaussian curvature as
\bq
 \kappa &=& \frac{R_{1212}}{g_{11}g_{22} - g_{12}^2}\,, \nn
 \eq
and from this we see that for any $(x,y)\in M$, we have\footnote{\bl{The formula for $\kappa(x,y)$ in \eqref{eq:kappa:GeneralModel} is also valid when $W^1$ and $W^2$ are correlated.}}
\bq
\label{eq:kappa:GeneralModel}
\kappa(x,y) &=& \frac{\alpha(y)(-2 \alpha(y) + y \alpha^\prime(y))}{y^2}. \,
\eq

(note that $\kappa  (x,y)$ does not depend on $x$ or $\sigma(x)$)
\sk
We now make the following additional assumptions:

\sk

\begin{Assumption}
\label{Assumption:TechnicalConditions}
\quad
\begin{itemize}
\item $\mu,\alpha,\sigma$ are $C^{\infty}$ and $\alpha,\sigma$ are strictly positive and $\al$ is strictly increasing and $\sigma \in C^2_b$, and $\mu$, $\al$ are such that \bl{$Y_t$ has a unique strong solution for any given $Y_0>0$ (e.g. Theorem 2.9 in  \cite{KS91}).}  And
 $\sigma$ is Lipschitz continuous.
\item $\alpha(y) \sim A_1 y$, $\alpha'(y) \sim A_1=\alpha'(0+)$, \bl{$\alpha''(y)\to 0$} as $y \to 0$, and $\alpha(y) \sim B_1 y^p, \alpha'(y)\sim B_1py^{p-1}, \bl{\alpha''(y)/y^{p-2}\to B_1p(p-1)}$ as $y \to \infty$ for some constants $A_1,B_1>0$ and $p \in (0,1]$.  This ensures that the associated Riemannian manifold on the upper half plane is complete - the distance to $y=0$ and $y=\infty$ under the metric $g_{ij}$ is infinite, and ensures that $y=0$ and $y=\infty$ are unattainable boundaries.  The condition at $y=\infty$ ensures that $X_t$ has a fatter (and thus more realistic) right tail than if we chose a bounded volatility function $f(y)$ as in  \cite{FJ11} (see  \cite{AP07},  \cite{Jour04},  \cite{LM07} for more details).

\item \bl{ For all sufficiently small $y>0$, $\mu(y)\ge 0$ and $\mu(1/y)\le 0$. Moreover, $\mu$ is such that, as $y\to0$,  $\overline{V}(y)$ and $\overline{V}(1/y)$ are bounded from above, where
\[\overline{V}(y):= \mu(y)g(y)+\half\alpha(y)^2 \left[g(y)^2+g'(y) \right],\text{ with }g(y):=-\frac{\mu(y)}{\alpha(y)^2}+\half\bigg(\frac{\alpha'(y)}{\alpha(y)}-\frac{1}{y}\bigg).\]
} These assumptions are needed to make the Gauge transformation trick with the Davies heat kernel estimate work.

\item $0<   \underline{\sigma} \le \sigma \le \bar{\sigma} < \infty$ for some constants $\underline{\sigma},\bar{\sigma}$.

\item We assume that $\sigma(x)^2 \,+\,\sigma'(x)^2\,-\,2\sigma(x)\sigma''(x)>0$ for all $x$, which is clearly true if $\sigma$ is constant.  This condition is required for the gauge transformation trick to work, and essentially just excludes excessive skew/convexity of the local volatility function $\sigma(x)$.

\end{itemize}
\end{Assumption}

\sk

\begin{prop}
\label{prop:EandU}
Under Assumption \ref{Assumption:TechnicalConditions}, the system of two-dimensional stochastic differential equations in \eqref{eq:Model} has a unique strong solution.
\end{prop}
\begin{proof}
Follows from standard arguments.
\end{proof}

\bs

\begin{Assumption}
\label{Assumption:NegCur}
 We assume that $-2 \alpha(y) + y \alpha^\prime(y)
   \le 0$ which implies that $\kappa(x,y)\le 0$ $\,\,\forall (x,y) \in M$, which (by Hadamard's theorem, see page 149 in  \cite{doC92}) implies that the cut locus of $M$ is empty. \end{Assumption}

\sk

\bl{\begin{rem}\label{rmk:410}
We list a few stochastic volatility (SV) models that satisfy both Assumption \ref{Assumption:TechnicalConditions} and Assumption \ref{Assumption:NegCur} in Table \ref{table:sv}.
\begin{table}[ht]
\caption{Examples of SV models}
\centering
\begin{tabular}{|c|c|c|}
\hline
\hline
Model & $\mu(y)$ & $\alpha(y)$\\ [0.5ex] 
\hline
SABR with $\beta=1$&0&$\nu y$ with $\nu>0$\\
\hline
Mean-reverting SABR with $\beta=1$&$\eta(\theta-y)$ with $\eta,\theta>0$&$\nu y$ with $\nu>0$\\
\hline
Exponential OU&$\eta(\theta-\log y)y$ with $\eta>0$, $\theta\in\mathbb{R}$& $\nu y$ with $\nu>0$\\
\hline
\end{tabular}
\label{table:sv}
\end{table}
\end{rem}
}

\begin{rem}
\label{rem:CurBoundBelow}
Using \eqref{eq:kappa:GeneralModel}, we see that $\kappa(x,y)$ is smooth for all $(x,y) \in M$ and independent of $x$, and
\begin{alignat*}{2}
\kappa(x,y) \sim& -A_1^2  \quad \quad\quad\quad &&(y \to 0), \,\nn \\
\kappa(x,y) \sim& - B_1^2 (2-p) y^{2(p-1)}\,\to \bl{\begin{cases}-B_1^2&\text{if }p=1\\0\quad &\text{if }p\in(0,1)\end{cases}} \quad \quad\quad &&(y \to \infty)\,, \nn \,
\end{alignat*}
 so $\kappa$ is \bl{bounded from  below}.
\end{rem}

\sk

\begin{rem}
Our conditions include the SABR model for $\beta=1$ (which corresponds to $p=1$) but not the Heston model, because for the latter the associated manifold is not complete, and completeness is needed for the Davies heat kernel estimate below.  Small-time asymptotics for the Heston model are obtained in  \cite{FJL12} using Fourier methods and saddle point estimates for contour integrals.
\end{rem}


\sk
\subsection{Tail behaviour of the model}

\sk
\sk

For the SABR model with zero correlation (i.e. $\al(y)=\al y$, $\mu(y)=0$), it is well known that for $m>0$, $\mathbb{E}(S_t^{m})<\infty$ if and only if $m \le 1$.  If $d\langle W^1, W^2\rangle_t=\rho dt$ with $\rho<0$, this condition changes to $\mathbb{E}(S_t^m)<\infty$ if and only $\rho \le -\sqrt{(m-1)/m}$ (see Theorem 2.3 in  \cite{LM07}).  Moreover, by Theorems 2.5 and 2.6 in  \cite{LM07}, this result also applies to our model in \eqref{eq:Model} if $p=1$ because $\al(y)$ does not have quadratic growth at $y=\infty$, i.e. the $b_{\infty}$ term in equation (29) in  \cite{LM07} is zero.  For $p<1$, the conditions are more complicated, we refer the reader to Theorem 3.2 in  \cite{LM07} for details.

\bs

\subsection{Example: the SABR model}

\sk
\sk

\sk
\sk

 For the hyperbolic metric $ds^2=\frac{1}{y^2}(dx^2+dy^2)$ on the upper half plane $\mathbb{H}^2$ (which is associated with the SABR model $dS_t=S_t^{\beta} Y_t dW^1_t, dY_t= Y_t dW^2_t$, $d W^1_t dW^2_t =\rho dt$ with $\beta=1,\rho=0$, see section \ref{section:SABR}), we have that $\kappa=-1$ (see e.g. Molchanov \cite{Mol75} or chapter 5 in  \cite{doC92}).  For $\mathbb{H}^2$ we have the simple explicit formula due to McKean \cite{McK70}
\bq
p_t(\mb{x},\mb{y})&=&\frac{\sqrt{2} \,e^{-t/8}}{(2\pi t)^{3/2}}\int_{d(\mb{x},\mb{y})}^{\infty} \frac{re^{-r^2/2t}}{\sqrt{\cosh r -\cosh d(\mb{x},\mb{y})}}dr\, \nn
\eq
(see e.g. Theorem 3.1 in Matsumoto \& Yor \cite{MY05} for the corresponding formula for the $n$-dimensional hyperbolic space).

\sk
\sk
\sk

\section{Small-time asymptotics for call options}


\bs

\subsection{A gauge transformation to remove the $\mc{A}$ term}
\label{subsection:Gauge}

\bs

Here and throughout, we let $p^0_t(\mb{x},\mb{y})$ denotes the heat kernel associated with the usual Laplace-Beltrami operator (i.e. with $\mc{A}=0$).

\sk

The following lemma computes an upper bound for $p_t(\mb{x},\mb{y})$ in terms of the heat kernel $p^0_t(\mb{x},\mb{y})$ for the case when $\mc{A}=0$.  This is needed so we can appeal to the Davies heat kernel estimate that follows.

\sk
\sk

\begin{lem}
We have the following upper bound for $\hat{p}_t(\mb{x},\mb{y})=\sqrt{|g(\mb{y}) |}\,p_t(\mb{x},\mb{y})$:
\label{lem:GaugeBound}
\bq
\label{eq:GaugeBound}
\hat{p}_t(\mb{x},\mb{y}) &\le&
   \frac{\chi(x_0,y_0)}{\chi(x_1,y_1)}\, e^{V_{\mathrm{max}}t} \,\hat{p}^0_t(\mb{x},\mb{y})\,,   \,
\eq
\nind for some constant $V_{\mathrm{max}}<\infty$, where $\mb{x}=(x_0,y_0)$, $\mb{y}=(x_1,y_1)$ and
\bq \chi(x,y) &=& \sqrt{\sigma(x)}\,e^{\half x} \frac{\sqrt{\al(y)}}{\sqrt{y}} \,e^{-\int_1^y \frac{\mu(u)}{\al(u)^2}du}\, \nn
\eq
\nind (see also pages 108-9 in  \cite{HL08} for related discussion).
\end{lem}
\nind \begin{proof}
If we set $(x,y)=(x_0,y_0)$, then we know that $\hat{p}_t(\mb{x},\mb{y})$ is a solution to the backward Kolmogorov equation
\bq
\label{eq:OriginalPDE}
\p_t \hat{p} &=& \left(\mc{A}\,+\,\half \Delta \right)\hat{p}  \,,
\eq
subject to $\hat{p}_0(\mb{x},\mb{y})=\delta(\mb{y}-\mb{x})$, where the spatial partial derivatives in $\Delta$ and $\mc{A}$ are with respect to the backward variable $\mb{x}=(x,y)$.  If we now let $\hat{p}_t(\mb{x},\mb{y})=\frac{\chi(x,y)}{\chi(x_1,y_1)} q_t(\mb{x},\mb{y})$, then the PDE transforms to
\bq\label{eq:TransformedPDE}
\p_t q&=&y^2\sigma(x)^2 \left(\frac{\p_x \chi}{\chi}-\frac{1}{2} \right) \p_xq+ \left(\mu(y)+\al(y)^2\frac{\p_y\chi}{\chi} \right) \p_y q
	+\half y^2\sigma(x)^2\p^2_xq+\half\al(y)^2\p^2_yq+V(x,y)\,q\nn\\
&=&\half \Delta  q \,+\, V(x,y) \,q\,
\eq
\nind with $q_0(\mb{x},\mb{y})=\delta(\mb{y}-\mb{x})$, where $V(x,y) =  \frac{(\mc{A}+\half \Delta) \chi}{\chi}(x,y)$.
In the second equality above, we used the fact that
\bq
\frac{\p_x\chi}{\chi}=\frac{1}{2}\bigg(1+\frac{\sigma'(x)}{\sigma(x)}\bigg)\equiv f(x),\quad \frac{\p_y\chi}{\chi}=-\frac{\mu(y)}{\alpha(y)^2}+\half\bigg(\frac{\alpha'(y)}{\alpha(y)}-\frac{1}{y}\bigg)= g(y) \,.\nn
\eq
In fact, one can use the above equations to obtain that
\bq
V(x,y)&=&\half y^2\sigma(x)^2 \left[-f(x)+f(x)^2+f'(x) \right]+\mu(y)g(y)+\half\alpha(y)^2 \left[g(y)^2+g'(y) \right]\nn\\
&=&-\frac{1}{8}y^2 \left[\sigma(x)^2+\sigma'(x)^2-2\sigma(x)\sigma''(x) \right]+\mu(y)g(y)+\half\alpha(y)^2 \left[ g(y)^2+g'(y) \right].\nn
\eq
From the final bullet point in Assumption \ref{Assumption:TechnicalConditions} we know that
\[V(x,y)\lee \overline{V}(y)\,\,\,\,= \,\,\,\, \mu(y)g(y)+\half\alpha(y)^2 \left[g(y)^2+g'(y) \right],\quad \forall (x,y)\in\R\times\R_+\,.\]
\bl{Moreover,  by Assumption \ref{Assumption:TechnicalConditions} we know that $\overline{V}(y)$ is uniformly bounded from above as $y\to0$ or $y\to\infty$.
 Thus, we know that $\overline{V}(y)$ is uniformly bounded from above for all $y\in(0,\infty)$ by some constant which we call $V_{\mathrm{max}}$ (which is independent of $x$),} and the same constant bounds $V(x,y)$ from above for all $(x,y)$.  Now let $\hat{\pr}$ denote the probability measure under which $(X_t,Y_t)$ defined in \eqref{eq:Model} has infinitesimal generator $L=\half \Delta$ (i.e. with $\mc{A}=0$). Then for any bounded Borel function $f$ over $\mathbb{R}\times\mathbb{R}_+$, and $\mb{x}=(x,y)$, $\mb{y}=(x',y')$, we have
\bq
\int f(\mb{y})q_t(\mb{x},\mb{y})d\mb{y}~ \ex_{\mb{x}}^{\hat{\pr}} \left(e^{\int_0^t V(X_s,Y_s)ds}f(X_t,Y_t) \right)
&=&\int f(\mb{y})\, \ex_{\mb{x}}^{\hat{\pr}} \left(e^{\int_0^t V(X_s,Y_s)ds}1_{(X_t,Y_t)\in d\mb{y}}\right)\,,\nn
\eq
where we have used Feynman-Kac formula in the first equality. By the arbitrariness of $f$, we have that
\bq
q_t(\mb{x},\mb{y})d\mb{y}=\ex_{\mb{x}}^{\hat{\pr}} \left( e^{\int_0^t V(X_s,Y_s)ds}1_{(X_t,Y_t)\in d\mb{y}} \right)\le e^{V_{\max} t}\hat{\pr}_{\mb{x}}\left( (X_t,Y_t)\in d\mb{y} \right) =e^{V_{\max} t}\hat{p}_t^0(\mb{x},\mb{y})d\mb{y}.\nn
\eq
It follows that
\bq
\hat{p}_t(\mb{x},\mb{y}) ~ \frac{\chi(x_0,y_0)}{\chi(x_1,y_1)} q_t(\mb{x},\mb{y})
&\le&   \frac{\chi(x_0,y_0)}{\chi(x_1,y_1)}\, e^{V_{\mathrm{max}}t} \hat{p}^0_t(\mb{x},\mb{y})  \nn \,.
\eq
\end{proof}

\sk

\sk
\begin{rem}
The gauge transformation trick here only works when the correlation $\rho=0$, unless we impose a specific functional form for $\mu(y)$, see subsection \ref{subsection:RhoNotZero} for details.  We would expect a similar result to Theorem \ref{thm:SmallTCallsOTM} to hold for $\rho\ne 0$ and general $\mu(y)$, but to prove this would require re-writing the whole of  \cite{Dav88} for the more general case when $\mc{A}\ne 0$ (\cite{Dav88} only deals with the self-adjoint case when $\mc{A}=0$).  Otherwise, we could impose that the volatility is given by $f(y)$ for some bounded function of $y$, and use the Norris-Stroock \cite{NS91} tail estimate for the fundamental solution to the heat equation with a uniformly elliptic coefficients instead of the Davies estimate, but this type of model would not have realistic fat tail behaviour.
\end{rem}

\sk
\sk

\begin{cor}
\label{cor:4.2}
\bl{Using that  $\alpha(y) \sim A_1 y$ as $y \to 0$ and $\alpha(y) \sim B_1 y^p$ as $y \to \infty$, we have a positive constant $\underline{C}>0$ such that:
\begin{enumerate}
\item For all sufficiently small $y>0$,
\bq
\chi(x,y) &\ge& \underline{C}\sqrt{A_1}\,\sqrt{\sigma(x)} \,e^{\half x} ; \nn
\eq
\item For all sufficiently large $y>0$,
\bq
\chi(x,y) &\ge &
\underline{C}\sqrt{B_1}\,\sqrt{\sigma(x)} \, e^{\half x} y^{-\half(1-p)}.\nn \,
\eq
\end{enumerate}
}
\end{cor}
\begin{proof}
\bl{This is an immediate consequence of the formula given for $\chi$ and our
assumptions on the asymptotics of $\mu$ and $\alpha$. Note that we
use our assumptions on the sign of $\mu$ near $0$ and $\infty$ to bound
$-\int_1^y \frac{\mu(u)}{\alpha(u)^2} d u$ from below.}
\end{proof}

\bs

\subsection{The Davies upper bound for the heat kernel}

\bs

By a simple modification of Theorem 16 in Davies \cite{Dav88} (which deals with the heat equation $\p_t u - \p^2_{xx} u=0$ without the $\half$ factor), we have the following:

\sk
\sk

\begin{thm}
\label{thm:Davies}
If $\mc{M}$ is a complete Riemannian manifold of dimension $N$ such that, for some constant $\beta\ge0$,
\bq
\mathrm{Ric}(\mb{x}) &\ge& -(N-1) \beta^2, \nn
\eq
where $\mathrm{Ric}$ denotes the Ricci curvature\footnote{The Ricci curvature is just equal to the Gaussian curvature when the dimension $n=2$, which is also equal to the sectional curvature.}, then there exists a constant $c_{\dl}$ depending on $\dl$ such that
\bq
0 \lee p^0_t(\mb{x},\mb{y}) \lee c_{\delta}\,|B(\mb{x},t^{\half})|^{-\half} |B(\mb{y},t^{\half})|^{-\half} \,e^{-d(\mb{x},\mb{y})^2/(2+\dl)t} \nn
\eq
for $0<t<1$, where $|B(\mb{x},r)|$ denotes the Riemannian volume of the ball $B(\mb{x},r)=\lb \mb{y} \in M: d(\mb{x},\mb{y})<r \rb$
(see also page 198 in  \cite{Chav84} for a similar result).
\end{thm}

\bs
\sk
\sk

We now return to our manifold $M$. Let $\mb{x}=(x_0,y_0)$, $\mb{y}=(x_1,y_1)$ denote two points on $M$ and let $d(x_0,y_0;x_1,y_1)=d(\mb{x},\mb{y})$.  From the assumption that $\kappa(x,y)\le0$
and the G\"{u}nther volume comparison theorem on page 213 in  \cite{Jost09}, we have that
\bq
|B(\mb{x},r)| &\ge & |B^E(\mb{x},r)| ~ \pi r^2 \,,
\eq
where $|B^E(\mb{x},r)|$ denotes the volume of a ball under the standard Euclidean metric.  Thus setting $r=t^{\half}$ we have the following corollary of Theorem \ref{thm:Davies}:

\sk
\sk

\begin{cor}
\bl{Using Remark \ref{rem:CurBoundBelow} and Theorem \ref{thm:Davies}}, we have the following upper bound
\label{cor:DaviesHKestimate}
\bq
 p^0_t(\mb{x},\mb{y}) \lee \frac{c_{\delta}}{\pi t} \,e^{-d(\mb{x},\mb{y})^2/(2+\dl)t}  \,,\nn
\eq
\nind which (combined with \eqref{eq:GaugeBound}) implies that
\bq
\hat{p}_t(\mb{x},\mb{y}) &\le&
   \frac{\chi(x_0,y_0)}{\chi(x_1,y_1)}\, e^{V_{\mathrm{max}}t} \,\sqrt{|g(x_1,y_1)|}  \frac{c_{\delta}}{\pi t} \,e^{-d(\mb{x},\mb{y})^2/(2+\dl)t}   \,\nn .
\eq

\end{cor}

\sk

\begin{lem}
\label{lem:Tails}
From a simple asymptotic analysis of a vertical line we have
\begin{alignat*}{4}
&d(\mb{x},\mb{y})&&\asymp d_0(y_1) \quad \quad  &&(y_1 \to 0) \,\,\,\,&&\text{where} \,\,\, d_{0}(y_1) :=  \bl{-}\frac{1}{ A_1}\log y_1 \,, \nn \\
&d(\mb{x},\mb{y}) &&\asymp d_\infty(y_1)     \quad   &&(y_1 \to \infty)\,\,\,&&\text{where} \,\,\, d_{\infty}(y_1) :=  \left\{
\begin{array}{ll}
\frac{1}{ B_1}\log y_1 , &(p=1)\,, \nn \\
\frac{y_1^{1-p}}{B_1 (1-p)},\quad&(p \in (0,1))
        \end{array}\
        \right.
 \,
\end{alignat*}
for $\mb{x}=(x_0,y_0)$ and $x_1$ fixed (where $\bl{\mb{y}}=(x_1,y_1))$.
\end{lem}

\sk
\sk

\subsection{Small-time expansion for non at-the-money call options}
\label{subsection:SmallTimeCalls}

\bs

In Theorem \ref{thm:SmallTCallsOTM}, we state the main result in the paper: a small-time expansion for out-of-the-money call options under the general local-stochastic volatility model in \eqref{eq:Model}.  To prove this result, we proceed along similar lines to section 3 in  \cite{Pau10}.  We introduce the following notation:
\bq
\phi(y_1)&=& \frac{1}{2}d(\mb{x},\mb{y})^2 \,,\nn \\
\psi(y_1)&=& y_1^2 \,\mc{P}(\mb{x},\mb{y}) \, u_0(\mb{x},\mb{y}) \sqrt{|g(x_1,y_1)|}\,,\nn
\eq
\nind where $\mc{P}(\mb{x}$, $\mb{y})=e^{A(\mb{x},\mb{y})}$, $\mb{x}=(x_0,y_0)$, $\mb{y}=(x_1,y_1)$.  Then the following theorem characterizes the small-time behaviour of non-at-the-money call options.

\sk
\sk

\begin{thm}
\label{thm:SmallTCallsOTM}
Consider the stochastic volatility model defined in \eqref{eq:Model} and assume the initial stock price $S_0=1$.\footnote{the result is easily adapted in the general case $S_0 \ne 1$ by considering $X_t=\log \frac{S_t}{S_0}$.}  Then we have the following small-time expansion for the price of a call option with strike $K \ne S_0$:
\bq
\mathbb{E}(S_t-K)^{+}\,-\,(S_0-K)^+ &\sim &
 \frac{A^{\mathrm{SV}}(x_1)}{\sqrt{2\pi}}\,e^{-\phi(y_1^*)/t}\,t^{\frac{3}{2}}\quad \quad \quad \quad (t \to 0)\,, \nn
\eq
\nind where $x_1=\log K$,
\bq
A^{\mathrm{SV}}(x_1)&=& \frac{K \sigma(x_1)^2\psi(y_1^*)}{\sqrt{\phi''(y_1^*)}}\,\frac{1}{2\phi(y_1^*)} \nn
\eq and $y_1^*=y_1^*(x_1)$ is the $y$-value where the shortest geodesic from $(x_0,y_0)$ hits the line $\lbrace x=x_1 \rbrace $ under the metric $g_{ij}$ in \eqref{eq:SVmetric}.
\end{thm}
\sk

\begin{rem}
\label{rem:Transversality}
Because the curvature $\kappa\le 0$, from the argument on page 209 in  \cite{doC92}, we know that there is a unique distance minimizing geodesic from $(x_0,y_0)$ to the line $\lb x=x_1\rb$, and we have the \textit{transversality} condition $\textbf{g}(\frac{d\gamma^*}{dt},(0,1))|_{(x_1,y_1^*)}=0$, i.e. the shortest geodesic comes in perpendicular to the vertical line under the metric $g_{ij}$ (see page 14 in  \cite{FJ11} for more details
on this point).  Moreover, because the correlation $\rho=0$, the shortest geodesic is also perpendicular in the usual Euclidean sense.
\end{rem}

\begin{rem}
There are semi-explicit formulae for computing $y_1^*(x)$ and $\phi(y_1^*)$; these are given by two integral equations in equations  (26) and (27) in the prequel article  \cite{FJ11}.  More specifically, we first solve for $y_1^*(x)$ numerically in equation (27) which is just a line search i.e. a one-dimensional root-finding exercise, and we then plug $y_1^*(x)$ into equation (26) to compute the distance of this shortest geodesic to the vertical line $\lb x=x_1 \rb $.
These calculations are quite low level and tedious, so we do not repeat these calculations in this article.
\end{rem}


\sk
\sk

\nind \begin{proof} Applying the generalized change-of-variable formula for semimartingales (Theorem 3.7.1, part (v) in Karatzas \& Shreve  \cite{KS91}) for $f(S)=(S-K)^+$ and using that $S$ is a martingale, we have
\bq
\mathbb{E} \left(f(S_t)\,-\,f(S_0) \right)&=&
\mathbb{E}\left(\int_{-\infty}^\infty \Lambda_t(S)\delta_K(dS)\right) ~ \mathbb{E} \left(\Lambda_t(K) \right)\,,\nn
\eq
where $\Lambda_t(a)$ is the semimartingale local time for $S_t$ at level $a$ and $\delta_K(dS)$ denotes the Dirac measure concentrated at $S=K$ (see also equation (3.6.47) in  \cite{KS91}). On the other hand, for any bounded, continuous function on $\mathbb{R}_+$, $g$, we have
\bq
\ex \left(\int_0^tg(S_u)d\langle S_u\rangle \right)
&=&\ex \left(\int_0^tg(S_u)S_u^2\sigma(X_u)^2Y_u^2du \right)\nn\\
&=&\int_0^t\ex(g(e^{X_u})e^{2X_u}\sigma(X_u)^2\ex(Y_u^2|X_u))du\nn\\
&=&\int_0^t\int_{-\infty}^\infty g(e^x)e^{2x}\sigma(x)^2\ex(Y_u^2|X_u=x)\pr(X_u\in dx)\,du\nn\\
&=&\int_{-\infty}^\infty g(e^x)e^{2x}\sigma(x)^2\int_0^t\ex(Y_u^2|X_u=x)\,\hat{p}_u(x_0,y_0;x)\,du \,dx,\label{eq:onehand}
\eq
where $\pr(X_u\in dx|X_0=x_0,Y_0=y_0)=\hat{p}_u(x_0,y_0;x)dx$. But by Theorem 3.7.1, part (iii) in Karatzas \& Shreve  \cite{KS91}, we have
\bq
\ex\left(\int_0^tg(S_u)d\langle S_u\rangle \right)
&=&2\,\mathbb{E}\left(\int_{0}^\infty g(S)\Lambda_t(S)dS \right) ~ 2\,\mathbb{E}\left(\int_{-\infty}^\infty g(e^x)\Lambda_t(e^x)e^xdx\right)\, .\label{eq:otherhand}
\eq
By the arbitrariness of $g$, comparing \eqref{eq:onehand} and \eqref{eq:otherhand} we see that
\bq
2e^x\ex \left(\Lambda_t(e^x) \right) ~e^{2x}\sigma(x)^2\int_0^t\ex(Y_u^2|X_u=x)\hat{p}_u(x_0,y_0;x)du ,\quad \quad \forall x\in\mathbb{R}\,.\nn
\eq
In particular
\be
\ex(f(S_t)-f(S_0)) ~\ex(\Lambda_t(K))~\ex(\Lambda_t(e^{x_1}))~\half K\sigma(x_1)^2\int_0^t\ex(Y_u^2|X_u=x_1)\hat{p}_u(x_0,y_0;x_1)du\,.\label{eq:TanakaMeyer}
\ee

\nind \bl{Let $1<a<\infty$ with $a>y_1^*$}.  Applying the Bellaiche heat kernel expansion on the compact interval $[\frac{1}{a},a]$, we know that for \bl{ any fixed $\e\in(0,1)$}, there exists a $t^*=t^*(\e)$ such that for all $t<t^*$ we have
\bq
 \ex(Y_t^2|X_t=x_1)\hat{p}_t(x_0,y_0;x_1)
&~& \int_{-\infty}^\infty y_1^2 \hat{p}_t(x_0,y_0;x_1,y_1) dy_1
\le    (1+\e) \, \int_{\frac{1}{a}}^a \frac{ \psi(y_1)}{2\pi t} \,e^{-\phi(y_1)/t}dy_1
\,+\, I_0  \,+\, I_{\infty},  \,\, \nn \\
\label{eq:Integrals}
&=&  (1+\e)\,\frac{\psi(y_1^*)}{\sqrt{2\pi t \,\phi''(y_1^*)}} \, e^{-\phi(y_1^*)/t}\left[1+O(t^{\half})\right]\,+\, I_0\,+\, I_{\infty}\,,
\eq
\nind where $I_0=\int_0^{\frac{1}{a}}y_1^2\hat{p}_t(x_0,y_0;x_1,y_1) dy_1$, $I_{\infty}=\int_a^{\infty}y_1^2\hat{p}_t(x_0,y_0;x_1,y_1) dy_1$, $y_1^*$ is defined in the statement of the theorem and we have used Laplace's method around the \mf{minimizer} at $y_1=y_1^*$ (see Proposition 2.1, page 323 in Stein \& Sharkarchi \cite{SS03}), which we are allowed to do because the distance function $d$, the metric $(g_{ij})$ and $u_0(\mb{x},\mb{y})$ are all smooth away from the cut locus of $x$ (and the cut locus is empty because $\kappa \le 0)$, so $\psi$ and $\phi$ are smooth.  Similarly we obtain the lower bound
\bq
 \ex(Y_t^2|X_t=x_1)\hat{p}_t(x_0,y_0;x_1)
&\ge &  (1-\e)\,\frac{\psi(y_1^*)}{\sqrt{2\pi t \,\phi''(y_1^*)}} \, e^{-\phi(y_1^*)/t}\, \left[1+O(t^{\half}) \right] \nn \, .
\eq

Fix $\delta>0$.  Then by Corollary \ref{cor:4.2} , Corollary \ref{cor:DaviesHKestimate} and \bl{Lemma \ref{lem:Tails}}, we know that for $a=a(\e)$ sufficiently large we have
\bq
I_0 &\le&    \int_0^{\frac{1}{a}}    y_1^2 \cdot \frac{\chi(x_0,y_0)}{\chi(x_1,y_1)}\, e^{V_{\mathrm{max}}t} \,\sqrt{|g(\mb{y})|} \,\, \frac{c_{\delta}}{\pi t} \,e^{-d(\mb{x},\mb{y})^2/(2+\dl)t}  dy_1 \nn \\
&\le & \frac{c_{\delta}}{\pi t} \,e^{V_{\mathrm{max}}t} \int_0^{\frac{1}{a}}    y_1^2  \sqrt{|g(\mb{y})|} \,\,  C_0 \, \,e^{-d_0(y_1)^2/C_0'(2+\dl)t}dy_1\,, \nn
\eq
\sk
for some positive constants $C_0=C_0(x_0,y_0,x_1)$ and $C_0'=C_0'(x_0,y_0,x_1)$, where $\sqrt{|g(\mb{y})|}\equiv\sqrt{|g(x_1,y_1)|}$.  But from the assumptions on $\al(y)$ we also know that $\sqrt{|g(\mb{y})|}\sim  \frac{1}{A_1\sigma(x_1) y_1^2}$ as $y_1 \to 0$.  \bl{Setting} $\varphi_{\e,\dl}=C_0'(2+\dl)(1+\e)$ and \bl{$c_{\e,\dl,t}= \frac{C_0}{A_1 \sigma(x_1)} c_{\dl}e^{V_{\mathrm{max}}t}(1+\e)/\pi$}, then we obtain
\bq
I_0 &\le&    \frac{c_{\e,\dl,t} }{ t}\, \int_0^{\frac{1}{a}}     \,e^{-d_0(y_1)^2/(\varphi_{\e,\dl}t)}dy_1\nn \\
&\le&   \frac{c_{\e,\dl,t} }{ t}   \int_0^{\frac{1}{a}}  \frac{1}{y_1}  \,e^{-[\frac{1}{A_1}\log y_1]^2/(\varphi_{\e,\dl}t)}dy_1 \nn \\
&=&    \frac{c_{\e,\dl,t} }{ t}   \int_{-\infty}^{\log\frac{1}{a}}    \,e^{-\frac{w^2}{A_1^2 \varphi_{\e,\dl}t}}dw \nn \\
&=&    \frac{c_{\e,\dl,t} }{ t}  \nu \sqrt{2\pi} \int_{-\infty}^{\log\frac{1}{a}} \frac{1}{\nu\sqrt{2 \pi}}\,e^{- \frac{w^2}{2\nu^2}}dw  \quad \quad \text{(where } \nu = A_1 \varphi_{\e,\dl}^{\half}\sqrt{t}/\sqrt{2}\,)\nn \\
&=&  \frac{c_{\e,\dl,t} }{ t}  \nu \sqrt{2\pi}\,\,\Phi(z) \lee   \frac{c_{\e,\dl,t}}{ t} \nu\frac{e^{-z^2/2}}{|z|}\,,  \quad \nn
\eq
\nind where $\Phi(z)=\int_{-\infty}^z \frac{1}{\sqrt{2\pi}}e^{-\half x^2}dx$, $z=\frac{\log \frac{1}{a}}{\nu}$, and we have used that $\frac{1}{y_1}\ge  \bl{1}$ \bl{for $0<y_1 <\frac{1}{a}<1$} and the inequality $\Phi(z) \le \frac{1}{|z|\sqrt{2\pi}}e^{-z^2/2}$ for $z<0$ in the last line.

\sk
\sk

  For $I_{\infty}$, \bl{we only consider the case that $p\in(0,1)$; the case of $p=1$ can be treated with an argument similar as the one  for $I_0$}. Again from Corollaries \ref{cor:4.2} and \ref{cor:DaviesHKestimate} and \bl{Lemma \ref{lem:Tails}}, we know that for $a=a(\e)$ sufficiently large we have
\bq
I_{\infty} &\le&    \int_a^{\infty}   y_1^2 \cdot \frac{\chi(x_0,y_0)}{\chi(x_1,y_1)}\, e^{V_{\mathrm{max}}t} \,\sqrt{|g(\mb{y})|} \,\, \frac{c_{\delta}}{\pi t} \,e^{-d(\mb{x},\mb{y})^2/((2+\dl)t)}  dy_1 \nn \\
 &\le & \frac{c_{\delta}}{\pi t}e^{V_{\mathrm{max}}t} \,\int_a^{\infty}    y_1^2 \,\,C_{\infty}  y_1^{\half(1-p)}
\, \sqrt{|g(\mb{y})|} \,\,  \,e^{-d_{\infty}(y_1)^2/(C_\infty'(2+\dl)t)}  dy_1 ,\nn \,
\eq
for some constants $C_{\infty}=C_{\infty}(x_0,y_0,x_1)$ and  $C_{\infty}'=C_{\infty}'(x_0,y_0,x_1)$.  But for $p \in (0,1)$,
$d_{\infty}(y_1) \mf{\sim} \frac{y_1^{1-p}}{ B_1 (1-p)}$ and $\sqrt{|g(\mb{y})|}\sim  \frac{1}{\sigma(x_1) B_1 y_1^{1+p}}$ as $y_1 \to \infty$.  Thus, we have
\bq
I_{\infty}  &\le&   \frac{\bar{c}_{\e,\delta,t}}{ t}\,\int_a^{\infty}y_1^{\frac{3}{2}(1-p)}\,e^{- \frac{y_1^{2(1-p)}}{ B_1^2 (1-p)^2}/\bl{(\bar{\varphi}_{\e,\delta}t})}dy_1 \nn \,
 \eq
where \bl{$\bar{c}_{\e,\dl,t}= \frac{C_{\infty}}{B_1 \sigma(x_1)} c_{\dl}e^{V_{\mathrm{max}}t}(1+\e)/\pi$ and $\bar{\varphi}_{\e,\delta}=C_\infty'(2+\delta)(1+\e)$}.  Making the change of variable $u=y_1^{1-p}$ in $I_{\infty}$ so $dy_1=\frac{1}{1-p} u^{\frac{p}{1-p}}du$, we see that for $a$ sufficiently large, we have
\bq
I_{\infty} &\le & \frac{\bar{c}_{\e,\delta,t}}{t}\int_{a^{1-p}}^{\infty} u^{\frac{3}{2}}\,e^{-\frac{u^2}{ B_1^2 (1-p)^2 \bl{\bar{\varphi}_{\e,\dl}}t }}\,\,\frac{u^{\frac{p}{1-p}} }{1-p }\, du \nn \\
  &=& \frac{\bar{c}_{\e,\delta,t}}{(1-p)t} \, \int_{a^{1-p}}^{\infty}   u^{\frac{3-p}{2(1-p)}}e^{- \frac{u^2}{ B_1^2 (1-p)^2\bl{\bar{\varphi}_{\e,\dl}}t}}\,du \nn \\
  &\le& \frac{\bar{c}_{\e,\delta,t}}{(1-p)t} \nu \sqrt{2\pi } \, \int_{a^{1-p}}^{\infty} \frac{1}{\nu \sqrt{2\pi }} \,  e^{- \frac{u^2}{2\nu^2}}du \nn \\
   &&  (\text{for } a=a(\e) \text{ sufficiently large, where } \nu =B_1 (1-p)\,\bl{\bar{\varphi}_{\e,\dl}^{\half}}\sqrt{t}\,)  \label{eq:matt-added} \\
  &=& \frac{\bar{c}_{\e,\delta,t}}{(1-p)t} \nu \sqrt{2\pi } \,\Phi^c(\frac{a^{1-p}}{\bl{\nu}}) \nn \\
  &\bl{\le}& \frac{\bar{c}_{\e,\delta,t}}{(1-p)t} \nu  \,\frac{e^{-\half z^2}}{z}\,,\nn
\eq
where $\Phi^c=1-\Phi$ and $z=\frac{a^{1-p}}{\nu}=  \frac{a^{1-p}}{ B_1 (1-p)\,\bl{\bar{\varphi}_{\e,\dl}^{\half}}\sqrt{t}}$.  Thus for $a$ sufficiently large, \mf{$z^2$ will exceed $\phi(y_1^*)/t$} so $I_0$ and $I_{\infty}$ are both higher order terms than the \bl{leading term $e^{-\phi(y_1^*)/t}/\sqrt{t}$} in \eqref{eq:Integrals}, thus we can ignore them at leading order.  A similar argument holds for the right tail integral $I_{\infty}$ when $p=1$.
 Thus we conclude that
\bq
\label{eq:InConclusion}
\ex(Y_t^2|X_t=x_1)\hat{p}_t(x_0,y_0;x_1)
&\sim &  \frac{\psi(y_1^*)}{\sqrt{2\pi t \,\phi''(y_1^*)}} \, e^{-\phi(y_1^*)/t}\,, \quad \quad \quad \quad (t \to 0) \,.
\eq

\sk
\sk
\nind We now have to estimate the integral in \eqref{eq:TanakaMeyer}, using \eqref{eq:InConclusion}.  Using the well known asymptotic relation
 \bq
 \label{eq:Awkward}
 \frac{1}{2} \int_0^t  \frac{1}{\sqrt{2\pi s   }}\,e^{-k^2/2s} ds &=&  \frac{1}{k^2\sqrt{2\pi}}\,t^{\frac{3}{2}} \, e^{-\frac{k^2}{2 t}}\,\,\left[1+O(t)\right] \,\quad \quad \quad \quad (t \to 0)\,,  \nn  \,
 \eq
\nind and comparing with \eqref{eq:TanakaMeyer} we see that for all $\e>0$ there exists a $t^*=t^*(\e)$ such that for all $t \le t^*$ we have
\bq
\mathbb{E}(S_t-K)^{+}\,-\,  (S_0-K)^+& \le & \frac{1}{2}K \sigma(x_1)^2 \int_0^t \frac{\psi(y_1^*)}{\sqrt{2\pi s \phi''(y_1^*)}} \,e^{-\phi(y_1^*)/s} (1+\e) ds \nn \\
&=& K\sigma(x_1)^2  \,\frac{\psi(y_1^*)}{\sqrt{\phi''(y_1^*)}}\frac{1}{2}\int_0^t \frac{1}{\sqrt{2\pi s }} \,e^{-\phi(y_1^*)/s}(1+\e) ds \nn \\
&=& K \sigma(x_1)^2 \,\frac{\psi(y_1^*)}{\sqrt{\phi''(y_1^*)}}\, \frac{1 }{2\phi(y_1^*)\sqrt{2\pi}}\,t^{\frac{3}{2}} \, e^{-\frac{\phi(y_1^*)
}{ t}}\,[1+O(t)]\,(1+\e) \nn \\
&\le& \frac{A^{\mathrm{SV}}(x_1)}{\sqrt{2\pi}}\,e^{-\phi(y_1^*)/t}\,t^{\frac{3}{2}}(1+2\e)\,\nn
\eq
(recall that $A^{\mathrm{SV}}(x_1)= \frac{K \sigma(x_1)^2\psi(y_1^*)}{\sqrt{\phi''(y_1^*)}}\,\frac{1}{2\phi(y_1^*)}$).  We proceed similarly for the lower bound.
\nind \end{proof}

\bs

\subsection{Non-zero correlation}
\label{subsection:RhoNotZero}

\bs
If $d\langle W^1, W^2\rangle_t=\rho dt$ for $\rho \ne 0$, we can still make the gauge transformation trick work, if $\mu(y)$ takes a certain functional form in terms of $\al(y)$ and $\sigma(x)$ is constant, as the following proposition demonstrates.  However, we assume that $\rho \le 0$ to ensure that the stock price process $S_t=e^{X_t}$ is a martingale (see e.g. \cite{Jour04}, \cite{LM07} to see examples of where this fails for $\rho >0$).

\sk
\begin{prop}
\label{prop:GaugeBoundRho}
For $\rho \neq 0,\pm1$, we can find a gauge transformation to remove the $\mc{A}$ term if $\sigma(x)$ is constant,  and
\bq
\mu(y) &=& \mf{\frac{\al(y)}{2y} \left[ y\al'(y)-\al(y) \right]}. \nn \,
\eq
Under this condition, \bl{and assuming $\mu$ and $\al$ also satisfy the other conditions in Assumption \ref{Assumption:TechnicalConditions}}, the potential $V(x,y)$ induced by the gauge transformation is bounded from above, so (if $\rho \le 0)$ Theorem \ref{thm:SmallTCallsOTM} still holds subject to minor modifications of the proof, using the following distance estimates for $\rho \ne 0$

\begin{alignat*}{4}
&d(\mb{x},\mb{y})&&\asymp d_0(y_1), \quad && (y_1 \to 0) \,\,\,\,&&\text{where} \,\,\, d_{0}(y_1) :=  \bl{-}\frac{1}{\bar{\rho} A_1}\log y_1,  \nn \\
&d(\mb{x},\mb{y}) &&\asymp d_\infty(y_1),    \quad \quad &&  (y_1 \to \infty)\,\,\,&&\text{where} \,\,\, d_{\infty}(y_1) :=  \left\{
\begin{array}{ll}
\frac{1}{\bar{\rho} B_1}\log y_1,  \quad \quad &(p=1)\,, \nn \\
\frac{y_1^{1-p}}{\bar{\rho}B_1 (1-p)}, \quad \quad&(p \in (0,1))
        \end{array}\
        \right.
 \,
\end{alignat*}
when $\sigma(x) \equiv 1$ and $\bar{\rho}=\sqrt{1-\rho^2}$.
\end{prop}
\nind \begin{proof}
See Appendix C.
\end{proof}

\sk
\sk

\begin{rem}
Setting $\al(y)=\nu y$ for $\nu>0$, we find that $\mu(y)=0$ which is consistent with the SABR model (for $\beta=1)$, so the gauge transformation works for this case.
\end{rem}

\sk
\sk

\subsection{Small-time behaviour for the Black-Scholes model}
\label{subsection:BlackScholes}

\bs

Let $C^{\mathrm{BS}}(S_0,K,t,\sigma)$ denote the price of a European call option at time zero under the Black-Scholes model with stock price $S_0$, strike $K=S_0 e^x$, time-to-maturity $t$, and volatility $\sigma$ (with zero interest rates and dividends).

\sk
\sk

\begin{prop}
\label{propn:BSSmallT}
Let
\bq \hat{\sigma}_t=\sqrt{\sigma^2+at} \nn
\eq
\nind for $t>0$, and assume $t \in (0,\frac{\sigma^2}{|a|})$ if $a<0$.  Then $C^{\mathrm{BS}}(S,K,t,\hat{\sigma}_t)$ has the following asymptotic behaviour as $t \to 0$
\bq
\label{eq:CBStimedep}
C^{\mathrm{BS}}(S,K,t,\hat{\sigma}_t)
&=& (S_0-K)^{+}\,+\, \frac{e^{-\frac{x^2}{2\sigma^2 t}}}{\sqrt{2\pi}}  t^{\frac{3}{2}} e^{\frac{1}{2}\frac{ax^2}{\sigma^4}} A_{\mathrm{BS}}(x,\sigma)[1+O(t)]  \quad \quad\quad \quad (x \ne 0),
\eq
\nind where $K=S_0e^x$ and $A_{\mathrm{BS}}(x,\sigma)=\mf{S_0 e^{\half x}\frac{\sigma^3}{x^2}}$.
\end{prop}
\begin{proof}
\mf{This is just the same as} Proposition \bl{3.4} in \cite{FJL12}.
\end{proof}

\sk

\begin{rem}
Note that we are not considering a time-dependent Black-Scholes model here per se, but rather a standard Black-Scholes model but with a volatility parameter which depends on $t$.
\end{rem}

\sk
\sk

\section{Small-time behaviour of implied volatility}

\bs

In this section we derive the corresponding asymptotic expansions for implied volatility.

\sk
\sk

\begin{thm}
\label{thm:ImpliedVol}
For the model defined above, let $\hat{\sigma}_t(x_1)$ denote the implied volatility at time zero and maturity $t$ for strike $K=e^{x_1}$, $K \ne S_0$.  Then we have the following asymptotic behaviour for $\hat{\sigma}_t(x_1)$:
\bq \label{eq:ImpliedVolExpansion}
\hat{\sigma}^2_t(x_1) &=& \hat{\sigma}^2(x_1)  +a(x_1) t\, \,+ \,o(t)\,,
\eq
\nind where
\bq \label{eq:sigmaa}
\hat{\sigma}(x_1)\,=\, \frac{|x_1-x_0|}{d(x_0,y_0;x_1,y_1^*(x_1))}\,,\,\quad \quad
a(x_1) \,=\, \frac{2\hat{\sigma}^4(x_1)}{x^2}\log\frac{A^{\mathrm{SV}}(x_1)}{A_{\mathrm{BS}}(x,\hat{\sigma}(x_1))}\,,
\eq
\nind  where $x=\log \frac{K}{S_0}$.
\end{thm}

\sk
\nind \begin{proof}
See Appendix \ref{section:HestonSmallTImpliedVol}.
\end{proof}

\sk

\begin{rem}
Note that we have not said anything about uniform convergence in $x$ in \eqref{eq:ImpliedVolExpansion}.  Typically for these type of problems, convergence is uniform on compact sets away from zero,
but the convergence is not uniform on intervals of the form $(0,a)$ or $(-a,0)$ for $a>0$ (but we do not need such a result in this article, so we defer the details for future work).
\end{rem}

\begin{rem}
From \eqref{eq:ImpliedVolExpansion} we see that
\bq
\hat{\sigma}_t(x_1) &=& \hat{\sigma}(x_1)  +\half \frac{a(x_1)}{\hat{\sigma}(x_1)} t\, \,+ \,o(t)\,.
\label{eq:ImpliedVolExpansion2}
\eq
\end{rem}
\bs


\section{Comparison against the asymptotic expansion in Lorig, Pagliarani  \&  Pascucci}

\bs
 \cite{LPP15} consider a general local-stochastic volatility model for a log stock price process $X_t$ of the form
\begin{align}
\left.
\begin{aligned}
d X_t
    &=  -\frac{1}{2}\sig^2(t,X_t,Y_t) d t + \sig(t,X_t,Y_t) d W_t , & X_0 &= x \in \Rb , \\
d Y_t
    &=  f(t,X_t,Y_t) d t + \beta(t,X_t,Y_t) d B_t , & Y_0 &= y \in \Rb, \\
d \< W, B \>_t
    &=  \rho(t,X_t,Y_t) \, d t , &
|\rho|
    &< 1 .
    \end{aligned}
    \right\}
    \label{eq:StochVol}
\end{align}
By expanding the coefficients of the infinitesimal generator of $(X,Y)$ in a Taylor series about an arbitrary point $(\bar{x},\bar{y})$, the authors obtain an explicit expansion for the price of a European call option and its associated implied volatility.  Under suitable conditions, the price of a European option $u(t,x,y):= \Eb_{t,x,y}(\varphi(X_T))$ \footnote{
In this section, $\Eb_{t,x,y}( \, \cdot \, )$ is shorthand for $\Eb( \, \cdot \, | X_t = x, Y_t = y)$, for all $(t,x,y)\in[0,T]\times \Rb \times \bl{\Rb}$.
} satisfies the backward Kolmogorov equation
\begin{align}
(\p_t + \Ac(t)) u
	&=	0 , \quad
u(T,x,y)
	=	\varphi(x) , \label{eq:kbe}
\end{align}
\bl{where $\Ac(t)$ is the infinitesimal generator associated with the two-dimensional process $(X,Y)$.}
We now briefly explain how the  \cite{LPP15} methodology works, in the one-dimensional case and for a general local-stochastic volatility model (see also  \cite{LPP14} and  \cite{PP14}).

\bs

\subsection{The one-dimensional case}

\bs

We first consider the one-dimensional heat equation
\bq
\left[ \p_t + a(x) \p^2_x \right] u &=& 0 \nn \,.
\eq
If we now formally expand $a(x)$ around zero: $a(x) = a(0) + a'(0) x + \half a''(0) x^2 + ...$ and set
$u = u_0 + u_1 + u_2 + ... $,  we obtain the following nested system of Cauchy problems:
\begin{align}
\left( \p_t + a(0) \p^2_x \right) u_0
	&=	0 , &
u_0(T,x)
	&=	\varphi(x) , \nn \\
\left( \p_t + a(0) \p^2_x \right)  u_1
	&=	- a'(0) x \p^2_x u_0 , &
u_1(T,x)
	&=	0 , &
\nn \\
\left( \p_t + a(0) \p^2_x \right)  u_2
	&=	- a'(0) x \p^2_x u_1 - \half a''(0) x^2 \p^2_x u_0 , &
u_2(T,x)
	&=	0 , &
\nn
\end{align}
and so on.  In general each of these equations can be solved recursively using Duhamel's principle, and  \cite{LPP15} give an explicit formula for $u_n$.
It is also helpful to consider an artificial parameter $\e \in (0,1]$.  We then set $a^\e(x) = a(0) + \e a'(0) x +  \half \e^2 a''(0) x^2$, $u^\e = u_0 + \e u_1 + \e^2 u_2 + \ldots $, and to obtain the family of equations above, we collect terms of like order in $\e$, and then finally set $\e=1$.
\bs

\subsection{The general case}

\bs

Now consider the generator $\Ac(t)$ associated with a general local-stochastic volatility model of the form in \eqref{eq:StochVol}:
\bq
\Ac(t)
    &=&  a(t,x,y) ( \d_x^2 - \d_x ) + f(t,x,y) \d_y + b(t,x,y) \d_y^2 + c(t,x,y) \d_x \d_y , \nn
\eq
where the functions $a$, $b$, $c$ are defined as
\begin{align}
a(t,x,y)
    &:=  \frac{1}{2}\sig^2(t,x,y) , &
b(t,x,y)
    &:=  \frac{1}{2}\beta^2(t,x,y) , &
c(t,x,y)
    &:=  \rho(t,x,y) \sig(t,x,y) \beta(t,x,y) . \nn
\end{align}
Expanding each function $\{a,b,c,f\}$ as a Taylor series about a fixed point $(\xb,\yb) \in \Rb^2$:
\begin{align}
\chi(t,x,y)
	&:=	\sum_{n=0}^\infty \chi_n(t,x,y) , &
\chi
	&=	\{ a,b,c,f\} ,\nn \\
\chi_n(t,x,y)
	&:=	\sum_{k=0}^n \chi_{n-k,k}(t) \cdot (x-\xb)^{n-k}(y-\yb)^k , &
\chi_{n-k,k}(t)
	&:=	\frac{1}{(n-k)!k!} \d_x^{n-k}\d_y^k \chi(t,\xb,\yb) , \nn
\end{align}
the generator $\Ac(t)$ can now be written formally as
\begin{align}
\Ac(t)
	&=	\sum_{n=0}^\infty \Ac_n(t) , &
\Ac_n(t)
	&=	a_n(t,x,y) ( \d_x^2 - \d_x ) + f_n(t,x,y) \d_y + b_n(t,x,y) \d_y^2 + c_n(t,x,y) \d_x \d_y \, .
	\nn
\end{align}
We now expand the unknown function $u$ as a series $u = \sum_{n=0}^{\infty} u_n$.  Inserting this expansion, as well as the expansion for $\Ac(t)$ into the Kolmogorov backward equation we again obtain a nested sequence of Cauchy problems
\begin{align}
(\d_t + \Ac_0(t)) u_0
	&=	0 , &
u_0(T,x,y)
	&=	\varphi(x) , \label{eq:nested1} \\
(\d_t + \Ac_0(t)) u_n
	&=	-\sum_{k=1}^n \Ac_k(t) u_{n-k}, &
u_n(T,x,y)
	&=	0 , &
n
	&\geq 1.  \label{eq:nested2}
\end{align}
The sequence $(u_n)$ can be solved explicitly, and a general expression for the $n$th term is given in Theorem 2.6 in  \cite{LPP15}.  For European call/put options, the expansion lends itself to an explicit implied volatility expansion (see Section 3 of  \cite{LPP15}), and the number of terms in the price and implied volatility expansion grows like $n!$.  As such, for practical purposes, one can only compute terms up to order $n=4$.  The advantage of the above method is that the $n$th order approximation (for both price and implied volatility) can be applied to any diffusion whose coefficients are $C^n$ in the spatial variables.  However, to prove the accuracy of the pricing approximation more stringent conditions (discussed below) must be enforced.

%
%
%

\bs

\subsection{Asymptotic error estimates}\label{sec:convergence}

\bs

In \cite{LPP15}, the authors prove the following bound for
the error introduced by replacing the exact European option price $u(t,x,y)$ with the $N$-th order approximation $\bar{u}_{N}(t,x,y)= \sum_{i=0}^N u_i(t,x,y)$.

\sk


\sk

\begin{cor}\label{cor:errorprice}
Consider a European option with payoff function $\varphi$.
Fix $\overline{T}>0$ and $(\bar{x},\bar{y}) = (x,y)$.
\bl{Suppose the coefficients $a(t,\cdot,\cdot)$, $b(t,\cdot,\cdot)$, $c(t,\cdot,\cdot)$ and $f(t,\cdot,\cdot)$ and all of their partial derivatives up to order $N+1$ are bounded by a positive constant $M<\infty$, uniformly with respect to $t \in [0,\overline{T}]$ and that
\begin{align}
\frac{1}{M} (\xi^2 + \eta^2) \leq a(t,x,y) \xi^2 + c(t,x,y) \xi \eta + b(t,x,y) \eta^2 \leq M (\xi^2 + \eta^2) , \quad \forall \, t \in [0,\overline{T}], \, x,y,\xi,\eta \in \Rb . \nn
\end{align}
}
Then for any $0<t<T\leq \overline{T},\varepsilon>0$ there exists a constant $C$ such that
 \begin{align} \label{eq:error_estimate_price}
 \left|u(t,x,y)-\bar{u}_{N}(t,x,y)\right|\leq C(T-t)^{\frac{N+1}{2}}\int_{\mathbb{R}^2}
 \Gamma^{M+\varepsilon}(t,x,y;T,x',y') \varphi(x') d x'd y', \qquad 0\leq t<T,\
 (x,y)\in\mathbb{R}^2 ,
\end{align}
where $\Gamma^{M+\varepsilon}(t,x;T,y)$ is the fundamental solution of the two-dimensional heat
operator
\begin{align}\label{eq:heatoperator}
H^{M+\varepsilon}=(M+\varepsilon) ( \partial^2_{x} + \partial^2_{y} ) + \partial_t ,
\end{align}
and the constant $C$ depends only on $M$, $N$, $\overline{T}$ and $\varepsilon$.
\end{cor}
\begin{proof}
\bl{See Corollary 4.6 of \cite{LPP15}.}
\end{proof}

\bs
\sk

For the case of a European call option with log-moneyness $x_1>0$, we can re-write the error bound as
\bq
|u(t,x,y)-\bar{u}_N(t,x,y)| &\le & C t^{\frac{N+1}{2}} \mathbb{E}(e^{\sigma W_t}-e^{x_1})^+ \nn \\
&\sim&  C t^{\frac{N+1}{2}} e^{-\frac{x_1^2}{2\sigma^2 t}} \,\frac{\sigma^3 e^{x_1}}{\sqrt{2\pi} \, x_1^2} t^{\frac{3}{2}} \quad \quad (t \to 0), \nn
 \eq
where $\sigma^2=M+\e$ and $W_t$ is a standard one-dimensional Brownian motion starting at zero.   But the error in our call option expansion is of the order
\bq
e^{-\phi(y_1^*)/t}\,o(t^{\frac{3}{2}})\quad \quad (t \to 0) \, .\nn
\eq
Thus we see that the  \cite{LPP15} result only gives a tighter error bound \bl{than} our asymptotic call option estimate in Theorem \ref{thm:SmallTCallsOTM} if
\bq
N &\ge &  O \left( \frac{1}{t \log \frac{1}{t}} \right) \left[ \phi(y_1^*) -\frac{x_1^2}{2\sigma^2} \right]  \quad \quad  (t \to 0), \nn \,
\eq
which will require computing $O\left( (\frac{1}{t \log \frac{1}{t}})! \right)$ terms in the  \cite{LPP15} expansion (recall that $\frac{1}{t \log \frac{1}{t}} \to \infty$ as $t \to \ja{0}$, but grows slower than $\frac{1}{t}$) (also note that $\phi(y_1^*) -\frac{x_1^2}{2\sigma^2}$ is positive because $\sigma$ is an upper bound of the diffusion matrix via the ellipticity condition in Assumption 4.1i in \cite{LPP15}).

\sk
\sk

\section{Numerical Example: the SABR model}
\label{section:SABR}

\bs

Consider the well known SABR model for $\beta=1$ with unit vol-of-vol:
 \bq
        \left\{
        \begin{array}{ll}
        dX_{t}\,=\,-\frac{1}{2}Y_t^2 dt\,+\,Y_t dW^1_t\,,\\
        dY_{t}\,\,=\,Y_t dW^2_t\,,
        \end{array}\
        \right. \label{eq:34}
        \eq
\nind with initial value $(X_0,Y_0)=(x_0,y_0)\in\mathbb{R}\times\mathbb{R}_+$ and independent standard Brownian motions $W^1$ and $W^2$.  The metric associated with this model is the hyperbolic metric $ds^2=\frac{1}{y^2}(dx^2+dy^2)$ on the upper half plane $\mathbb{H}^2$, and $g_{11}(x,y)=g_{22}(x,y)=\frac{1}{y^2}$ and $g_{12}(x,y)=g_{21}(x,y)=0$, so $\sqrt{|g(x,y)|}=\frac{1}{y^2}$.  For the hyperbolic metric, it is known (see e.g. page 170 in \cite{HL08} and Paulot \cite{Pau10}\bl{)} that
\bq
d(x_0,y_0;x_1,y_1)&=& \cosh^{-1}[1+\frac{(x_1-x_0)^2+(y_1-y_0)^2}{\bl{2y_0 y_1}}] \,,\nn \\
y_1^*&=& \sqrt{(x_1-x_0)^2+y_0^2}\,,\nn \\
d:=d(x_0,y_0;x_1,y_1^*)&=&  \mf{\cosh ^{-1}[\frac{\sqrt{(x_1-x_0)^2+y_0^2}}{y_0}]} \, , \nn \\
E_{yy}(x_0,y_0;x_1,y_1^*)&=&  \phi''(y_1^*) ~ \frac{d(x_0,y_0;x_1,y_1^*)}{y_0 y_1^* \sinh d(x_0,y_0;x_1,y_1^*)}\,. \nn
\eq
\nind We also have
\bq
\Delta &=&\frac{1}{y^2} \left(\pdd{^2}{ x^2}+\pdd{^2}{ y^2}\right)\,,
\quad \mathcal{A}= -\frac{1}{2}y^2 \p_x\,, \quad u_0(x_0,y_0;x_1,y_1)= \left(\frac{\sinh d(x_0,y_0;x_1,y_1)}{d(x_0,y_0;x_1,y_1)}\right)^{-\frac{1}{2}}\,,\nn \\
A(x_0,y_0;x_1,y_1)&=& \int_0^1 \langle \mathcal{A},\dot{\gamma}\rangle \, dt ~ \int_{x_0}^{x_1} \frac{1}{y^2} \cdot -\frac{1}{2}y^2 \, dx ~ - \frac{1}{2}(x_1-x_0)\,.
\nn
\eq

\nind Without loss of generality, we can set $x_0=0$, and we obtain
\bq
A^{\mathrm{SV}}(x_1) &=& \frac{K \psi(y_1^*)}{\sqrt{\phi''(y_1^*)}}\,\frac{1}{2\phi(y_1^*)} ~\frac{Ke^{-\frac{1}{2}(x_1-x_0)} (\frac{\sinh d}{d})^{-\frac{1}{2}}}{\sqrt{d/(y_0 y_1^* \sinh d)}}\,\frac{1}{d^2}\,. \nn
\eq

\begin{figure}
\centering
\includegraphics[width=240pt,height=230pt]{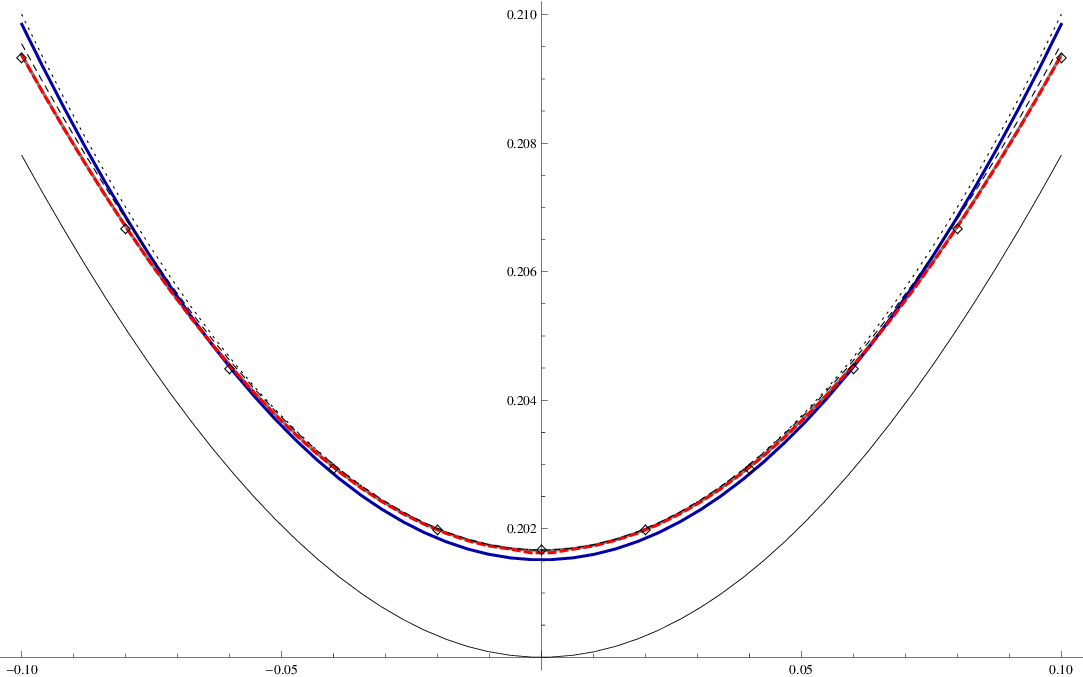}
\includegraphics[width=120pt,height=80pt]{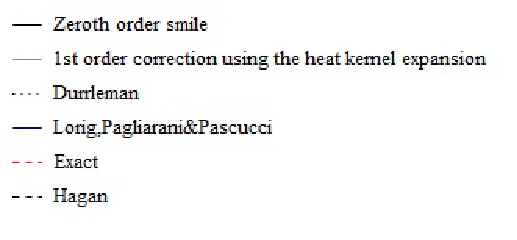} \\
\nind \caption{Here we consider the uncorrelated SABR model $dS_t=S_t Y_t dW^1_t, dY_t= \al Y_t dW^2_t $ with $S_0=1$, $y_0=.2$, $\alpha=1$ and maturity $t=.1$ - we have plotted the leading order smile (the lower curve, solid grey thin line) and the corrected implied volatility smile (light grey) using \eqref{eq:ImpliedVolExpansion}. We have also plotted the exact smile (computed via numerical integration using the formula in  \cite{AS12}) in red-dashed (which is barely distinguishable from our corrected smile in light grey), the Lorig, Pagliarani \& Pascucci truncated expansion with $N=3$ (dark blue), the smile using a fast Monte Carlo scheme (points with diamond-shaped markers, using 6 million simulations and $1000$ time steps), the smile using the well known Hagan et al. \cite{HKLW02} formula (black dashed, which lies above the red-dashed curve and below the LPP smile) and the Durrleman \cite{Dur04} approximation (in black dots) (cf. Theorem 3.1.1 in  \cite{Dur04}, which also agrees with Theorem 4.1 in  \cite{FJ11})(Mathematica code available from MF on request)}
\bs
\includegraphics[width=220pt,height=165pt]{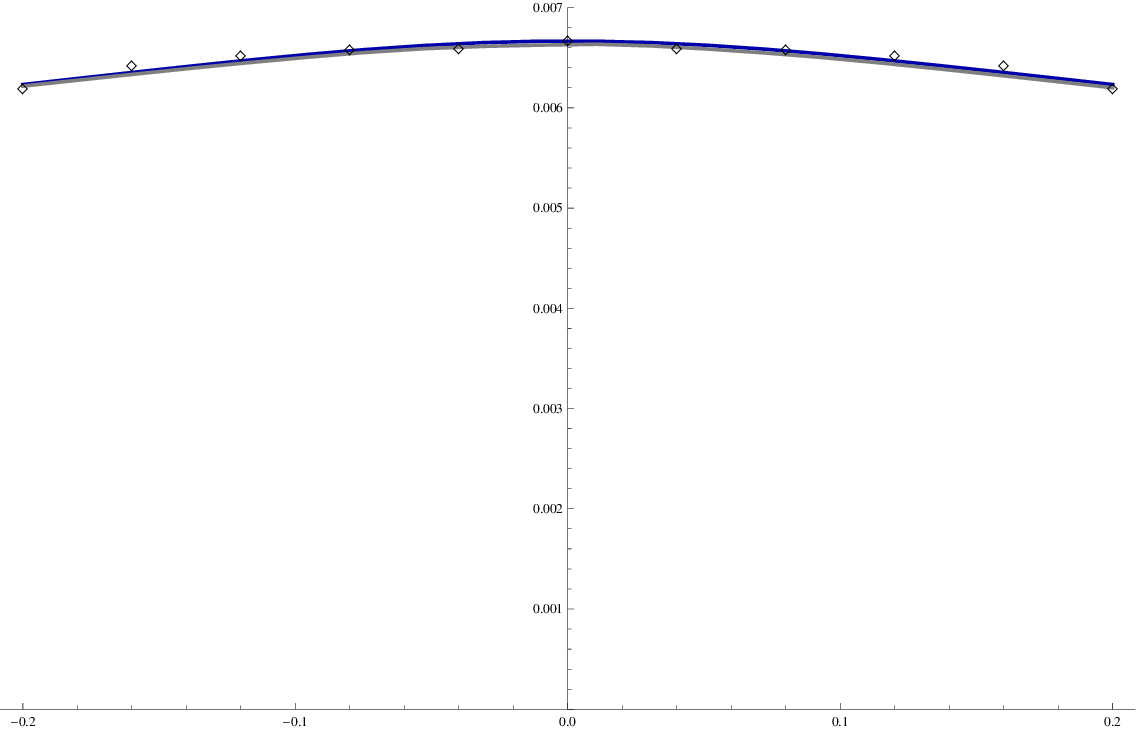}
 \caption{Here we have plotted the our correction term $a(x)$ (blue) against the $a(x)$ implied by \cite{AS12} and the $a(x)$ implied by the Monte Carlo call prices, i.e. $(\sigma_{MC}^2(x)-\sigma_0(x)^2)/t$ (diamond-shaped markers).  Our formula for $a(x)$ is also in exact (i.e. to machine precision) numerical agreement with the expansion given in section 4.3 in Paulot \cite{Pau10}, and in equation (6.10) in Busca et al. \cite{BBF2}.}
\end{figure}

\bs
In the first table below, we show the leading order smile, the first order corrected smile and the implied volatility from Monte Carlo $\sigma_{MC}$, as a function of the log-moneyness.  In the second table, we show $a(x)$ against the $a(x)$ implied by the Monte Carlo implied volatility given by $(\sigma_{MC}^2(x)-\sigma_0(x)^2)/t$ and also the $a(x)$ implied by the formula in \cite{AS12}.

\bs

\begin{center}
  \begin{tabular}{| l | l | l | l | l |}
    \hline
   $x$ & Zeroth order smile & First order correction & Monte Carlo &Relative Error \\ \hline
                       0.0001 & 0.2 & 0.201668 & 0.201664 & -0.00212726\%\\ \hline
                    0.04  & 0.201319 & 0.202961 & 0.20295 & -0.00542944\%\\ \hline
                        0.08 & 0.20511 & 0.206705   & 0.206709 & 0.00220052\%\\ \hline
   0.12 & 0.210961 & 0.212489 &  0.212503 & 0.00685656\%\\ \hline
   0.16 &0.21838 &  0.21983   & 0.219846 & 0.00760195\%\\ \hline
 0.2  &0.226919  & 0.228288 & 0.22828 & -0.00324944\% \\ \hline
  \end{tabular}
\end{center}

\bs

\begin{center}
  \begin{tabular}{| l | l | l | l |l |}
    \hline
   $x$ &  $a(x)$ & $[\sigma_{MC}^2(x)-\sigma_0(x)^2]/t$ & $a(x)$ Antonov \& Spectre & Relative error of $a(x)$ to MC\\ \hline
                       0.0001  & 0.00670034 & 0.00668307 & 0.00662428  & -0.257743 \%\\ \hline
                    0.04   & 0.00664065 & 0.00659592 &0.00659901& -0.673579\%\\ \hline
                        0.08  & 0.00656944   & 0.00658825 & 0.00653502 & 0.286241\%\\ \hline
                       0.12  &0.00646856 & 0.00653048 & 0.00643458& 0.957229\%\\ \hline
                    0.16   & 0.00635304 & 0.00642651 & 0.00632072&  1.15654\%\\ \hline
                        0.20  & 0.00623259   & 0.00619872 & 0.00620046 &-0.54341\%\\ \hline
  \end{tabular}
\end{center}

\bs

\newpage

\section{Small-time asymptotics for stochastic volatility model with a single jump-to-default}

 \bs

 We now extend the model in \eqref{eq:Model} by incorporating a single (independent) Poisson jump-to-default with constant intensity $\lambda>0$, and set $\rho=0$, $\sigma(x)\equiv 1$ for simplicity, so $S_t=e^{X_t}1_{\tau>t}$ where $X$ satisfies \bq
        \left\{
        \begin{array}{ll}
        dX_{t}\,=\, (\lambda-\frac{1}{2}Y_t^2) dt+Y_t
dW^1_t\,,\\
        dY_{t}\,\,=\,\mu(Y_t)dt+  \alpha(Y_t) dW^2_t\,,
        \end{array}\
        \right. \label{eq:JumpModel}
        \eq
  and $\tau \sim \mathrm{Exp}(\lm)$ is the default time for $S$ (which is independent of $W^1,W^2$), with $X_0=0$, $Y_0=y_0\in\mathbb{R}_+$.  Then the price of a out-of-the-money call option at time zero with strike $K>S_0$ is $e^{-\lm t} \mathbb{E}(e^{X_t}-K)^+$.  But $e^{-\lm t} =1 +O(t)$, so the effect of the actual default on the call option will not be seen at the order that we are interested in.  However the effect of the compensator drift term $\lm dt$ \textit{will} be felt at leading order, and will increase $\mc{P}(\mb{x},\mb{y})$ as defined in Section \ref{subsection:SmallTimeCalls} (and thus the call option price at leading order) by the following factor:
\bq
e^{\lm \int_0^1 \frac{1}{y(x(t))^2} \frac{dx}{dt}  \,dt} \,
\label{eq:LambdaExp}
\eq
where $x(t)$ is the $x$ coordinate of the shortest geodesic joining $(0,y_0)$ to the vertical line $\lb x=x_1 \rb$.  By a standard property of geodesics (see e.g. \cite{doC92} and also equation (16) in  \cite{FJ11}), we know that the speed of the geodesic
is conserved, i.e.
\bq
\label{eq:Lagrangian}
L ~ \frac{1}{y^2}\left(\frac{dx}{dt}\right)^2\,+\,\frac{1}{\al(y)^2}\left( \frac{dy}{dt} \right)^2 ~ E
\eq
for some Energy constant $E>0$, where $L$ is the Lagrangian.  But from the second Euler-Lagrange equation, we also know that
\bq
\frac{d}{dt} \left(\pdd{L}{\dot{x}} \right) ~ \frac{d}{dt} \left( \frac{1}{y^2}\frac{dx}{dt} \right) ~ 0\,.
\eq
Thus $K_1=\frac{1}{y^2}\frac{dx}{dt}$ is also a conserved quantity, so we can re-write \eqref{eq:Lagrangian} as
\bq
\label{eq:TurningPoint}
y^2 K_1^2\,+\,\frac{1}{\al(y)^2}\Big(\frac{dy}{dt}\Big)^2 ~ E \,.
\eq
We know from standard properties of geodesics that the shortest distance from $(0,y_0)$ to the line $\lb x=x_1\rb$ is $d(x_1)=\sqrt{E}$ and the shortest geodesic from $(0,y_0)$ to $\lb x=x_1\rb$ is perpendicular to the $y$-axis at $(x_1,y_1^*(x_1))$ so $\frac{dy}{dt}=0$ at this point, hence from \eqref{eq:TurningPoint} we have that $(y_1^*(x_1))^2=E/K_1^2$.  Thus
\bq
K_1&=&\frac{\sqrt{E}}{y_1^*(x_1)} ~ \frac{d(x_1)}{y_1^*(x_1)}\,,
\eq
so \eqref{eq:LambdaExp} simplifies to $e^{\lambda d(x_1)/y_1^*(x_1)}$, and from \eqref{eq:sigmaa} we find that the modified (i.e. jump-adjusted) correction term is given by
\bq
a^J(x) &=& a(x) \,+\, \frac{2 \lm \hat{\sigma}^4(x)}{x^2} \frac{d(x)}{y_1^*(x)} \nn  \,.
\eq
Then, using \eqref{eq:ImpliedVolExpansion2}, we see that the presence of the jump-to-default increases the implied volatility by the amount:
\bq \Delta \hat{\sigma}^{J}(x) \,\,\,\,:=\,\,\,\,  \lm \frac{\hat{\sigma}^3(x)}{x^2} \frac{d(x)}{y_1^*(x)} t \,+\,o(t) &=& \lm \frac{\hat{\sigma}^2(x)}{|x| \, y_1^*(x)} t \,+\,o(t) \nn \,.\eq

For an out-of-the-money put option with strike $K<S_0$, the small-time behaviour is qualitatively completely different, and we see that
\bq
\mathbb{E}(K-S_t)^{+} &=&\, K (1-e^{-\lm t})\,\,+\,\, e^{-\lm t } const. \times\,e^{-\phi(y_1^*)/t}\,t^{\frac{3}{2}}[1+o(1)] \quad \quad \,\, (t \to 0) \nn \\
 &=&\, K \sum_{n=1}^{\infty}\bl{\frac{(-1)^{n+1} (\lm t)^n}{n!}}\,+\, o(t^q) \quad \quad \quad \quad (t \to 0)
\label{eq:Jose2}
\eq
for any $q>0$.
From these observations, we note the following:
\begin{itemize}
\item
For $x>0$, the new correction term $a^J(x)$ tends to infinity as $x \to 0$, which means that $\lim_{x\searrow 0} \pdd{}{t}\hat{\sigma}^2_t(x)|_{t=0}=\infty$; this may not equal $\pdd{}{t} \hat{\sigma}_t(0)$ (which is not computed in this article), but
for a general exponential L\'{e}vy model with a non-zero Brownian component, equation (1.14) in Figueroa-L\'{o}pez et al. \cite{FGH14} show that
\bq
\hat{\sigma}_t(0) &=&\sigma_0 \,+\, const. \times t^{1-\frac{Y}{2}}+ o(t^{1-\frac{Y}{2}})\,,
\label{eq:Jose}
\eq
where $Y \in (0,2)$ measures the degree of jump activity, and \eqref{eq:Jose} implies that $\pdd{}{t} \hat{\sigma}_t(0)|_{t=0}=\infty$.

\item
The $K<S_0$ case is not so mathematically interesting for us; we see that the small-time put price expansion in \eqref{eq:Jose2} is essentially a Taylor series in powers $t$ (none of whose coefficients are affected by the stochastic volatility).  In this case, it can be shown that the implied volatility tends to $\infty$ as $t \to 0$ (see e.g. \cite{FFL12} et al. for details), similar to a pure exponential L\'{e}vy model; we refer the reader to \cite{FGH12} et al. for more results in this direction.
\end{itemize}

\bs

\sk
\sk

\appendix
\renewcommand{\theequation}{A-\arabic{equation}}
\setcounter{equation}{0}  

\section{Proof of Theorem \ref{thm:ImpliedVol}}
\label{section:HestonSmallTImpliedVol}

\sk
\sk

\ja{
The formulae given for $\hat{\sigma}$ and $a$ are chosen such that
$\phi^{\mathrm BS}(x,\hat{\sigma})=\phi(y^*)$ and $\hat{A}^{\mathrm{BS}}(x,\hat{\sigma},a) = A^{SV}(x_1)$, \mf{where $\phi^{\mathrm BS}(x,\sigma):=\frac{x^2}{2\sigma^2}$ and $\hat{A}^{\mathrm{BS}}(x,\sigma,a):=A_{\mathrm{BS}}(x,\sigma) e^{\frac{1}{2}\frac{ax^2}{\sigma^4}} $}.  Given $\delta>0$,
choose $\e < \frac{1}{4} \frac{x^2 \delta}{\sigma^4}$. From Proposition
\ref{propn:BSSmallT} and Theorem \ref{thm:SmallTCallsOTM}, there exists
a $t^*>0$ such that for all $t<t^*$
\bq
C^{\mathrm{BS}}(S_0,K,t,\sqrt{\hat{\sigma}^2 \,+\, t(a+\delta)}) - (S_0 - K)^+
&\geq& \frac{1}{\sqrt{2 \pi}} \hat{A}^{\mathrm{BS}}(x,\hat{\sigma}, a + \delta) e^{-\phi^{\mathrm{BS}}(x,\mf{\hat{\sigma}})/t} t^{\frac{3}{2}} e^{-\e}, \nn \\
\mathbb{E}(S_t-K)^{+}\,-\,(S_0-K)^+
&\leq& \frac{1}{\sqrt{2 \pi}} A^{\mathrm{SV}} e^{-\phi(y^*)/t} t^{\frac{3}{2}} e^\e
= \frac{1}{\sqrt{2\pi}} \hat{A}^{\mathrm{BS}}(x,\hat{\sigma}, a) e^{-\phi^{\mathrm{BS}}(x,\mf{\hat{\sigma}})/t} t^{\frac{3}{2}} e^{\e}. \nn
\eq
We now observe that $\e$ has been chosen so that
\bq
\frac{\hat{A}^{\mathrm{BS}}(x,\hat{\sigma}, a + \delta)}{\hat{A}^{\mathrm{BS}}(x,\hat{\sigma},a)}
~ e^{\frac{1}{2}\frac{x^2 \delta}{\hat{\sigma}^4}} \gt e^{2 \e}.
\eq
So
\[
\mathbb{E}(S_t-K)^{+}\,-\,(S_0-K)^+
< \frac{1}{\sqrt{2\pi}} \hat{A}^{\mathrm{BS}}(x,\hat{\sigma}, a + \delta) e^{-\phi^{\mathrm{BS}}(x,\mf{\hat{\sigma}})/t} t^{\frac{3}{2}} e^{-\e}
\leq
C^{\mathrm{BS}}(S,K,t,\sqrt{\hat{\sigma}^2 \,+\, t(a+\delta)}) - (S_0 - K)^+ .
\]
By the monotonicity of $C^{\mathrm{BS}}$ as a function of volatility, we deduce
\bq
\hat{\sigma}_t^2(x_1) \lee \hat{\sigma}(x_1)^2 + a(x_1) t + \delta t.
\eq
We proceed similarly to prove a lower bound.
}

\sk

\sk

\renewcommand{\theequation}{B-\arabic{equation}}
\setcounter{equation}{0}  

\section{Calculating $A(\mb{x},\mb{y})$ explicitly}
\label{section:CalculatingA}

\sk
\sk
Recall that
\bq
\mc{A} &=&  \left[-\frac{1}{2}y^2-\half  y^2 \sigma'(x)\sigma(x)\right] \p_x\,+\, \left[\mu(y) \,-\, \half (\al'(y) \al(y)  -\frac{\al(y)^2}{y})\right] \p_y \nn \,.
\eq
\nind
\ja{We calculate that}
\bq
A(x_0,y_0;x_1,y_1^*) &=&  \int_0^1 \langle \mathcal{A},\dot{\gamma}\rangle \,dt \nn \\
&=& \ja{\int_0^1} \left[\frac{1}{\sigma(x)^2 y^2}\mathcal{A}^1  \frac{\ja{d \gamma^1}}{dt}  \,+ \, \frac{1}{\alpha(y)^2}\mathcal{A}^2 \frac{\ja{d \gamma^2}}{dt}\right] dt \nn \\
&=& \ja{\int_0^1} \left[\frac{1}{\sigma(x)^2 y^2}(-\frac{1}{2}y^2-\half  y^2 \sigma'(x)\sigma(x))  \frac{\ja{d \gamma^1}}{dt}  \,+ \, \frac{1}{\alpha(y)^2}\mathcal{A}^2 \frac{\ja{d \gamma^2}}{dt}\right] dt \nn \\
&=&\ja{\oint_\gamma} \left( - \frac{1}{2 \sigma(x)^2}\left[1+\sigma'(x)\sigma(x)\right] dx \,+ \frac{1}{\al(y)^2}\left[\mu(y) \,-\, \half (\al'(y) \al(y)  -\frac{\al(y)^2}{y})\right] dy\,\right).\nn
\eq
This is the integral of an exact 1-form so its value does not depend upon the end points of $\gm$. We deduce that
\[
A(x_0,y_0;x_1,y_1^*)
=-\frac{1}{2}\int_{x_0}^{x_1} \frac{1}{\sigma(x)^2}\left[1+\sigma'(x)\sigma(x)\right] dx \,+ \,\int_{y_0}^{y_1^*} \frac{1}{\al(y)^2}\left[\mu(y) \,-\, \half (\al'(y) \al(y)  -\frac{\al(y)^2}{y})\right] dy\,.\nn
\]

\bs

\renewcommand{\theequation}{C-\arabic{equation}}
\setcounter{equation}{0}  

\section{Proof of Proposition \ref{prop:GaugeBoundRho}}

\sk

Denoting $\hat{p}(x_1,y_1,t)=\hat{p}_t(x_0,y_0;x_1,y_1)$, $q_t(x_1,y_1,t)=q_t(x_0,y_0,x_1,y_1)$, and substituting $\hat{p}(x,y,t)=h(x,y)\,q(x,y,t)$ into the original PDE \eqref{eq:OriginalPDE},
we need to find $h(x,y)$ such that the coefficients of $\partial_x q$ and $\partial_y q$ are $\half y^2\sigma'(x)\sigma(x)$ and $\frac{\alpha(y)}{2y}[y\alpha'(y)-\alpha(y)]$ respectively, i.e. agree with $\half \Delta$ in \ja{the equivalent expression to \eqref{eq:LBoperator}
for the case when $\rho\neq0$}.  Performing this substitution, we obtain
\bq
\p_t q&=&-\half y^2\sigma(x)^2 \left(\frac{\p_xh}{h}\, q+\p_x q\right)+\half y^2\sigma(x)^2 \left(\frac{\p^2_xh}{h}q+2\frac{\p_xh}{h}\p_xq+\p^2_xq\right)+\mu(y)
\left(\frac{\p_yh}{h}q+\p_yq\right)\nn\\
&&+\half\al(y)^2 \left(\frac{\p^2_yh}{h}q+2\frac{\p_yh}{h}\p_yq+\p^2_yq \right)+\rho y\sigma(x)\al(y)
\left(\frac{\p_x \p_yh}{h}q+\frac{\p_xh}{h}\p_y q+\frac{\p_y h}{h}\p_xq+\p_x \p_yq\right).\nn
\eq
Collecting coefficients of $q$ and its derivatives, we have
\bq
\p_t q&=&y^2\sigma(x)^2\left(\frac{\p_x h}{h}-\frac{1}{2}+\frac{\rho\al(y)}{\sigma(x)y}\frac{\p_y h}{h} \right) \p_xq+
\left(\mu(y)+\al(y)^2\frac{\p_yh}{h}+\rho y\sigma(x)\al(y)\frac{\p_xh}{h}\right)\p_y q\nn\\
&&+\half y^2\sigma(x)^2\p^2_xq+\half\al(y)^2\p^2_yq+\rho y\sigma(x)\al(y)\p_x \p_y q+V(x,y)q,\nn
\eq
where
\bq
V(x,y) &=&  \frac{(\mc{A}+\half \Delta) h}{h}\,.
\label{eq:Veexy}
\eq
Thus, to determine the function $h$, we impose that
\bq
y^2\sigma(x)^2 \left(\frac{\p_x h}{h}-\frac{1}{2}+\frac{\rho\al(y)}{\sigma(x)y}\frac{\p_y h}{h}\right)&=&\half y^2\sigma'(x)\sigma(x),\label{Aeq1}  \\
\mu(y)+\al(y)^2\frac{\p_yh}{h}+\rho \sigma(x)y\al(y)\frac{\p_xh}{h}&=&\frac{\al(y)}{2y}[y\al'(y)-\al(y)]\,.\label{Aeq2}
\eq
From \eqref{Aeq1} we obtain that $\p_xh/h=\half(1+\frac{\sigma'(x)}{\sigma(x)})-\frac{\rho\al(y)}{\sigma(x)y}\p_yh/h$. Plugging this into \eqref{Aeq2}, we have
\bq
\frac{\p_yh}{h} ~\frac{1}{1-\rho^2} \left[\frac{\al'(y)}{2\al(y)}-\frac{1}{2y}-\frac{\mu(y)}{\al(y)^2}-\frac{\rho y}{2\al(y)}(\sigma(x)+\sigma'(x))\right] ~ g(x,y) \,, \nn
\eq
so \bq
h(x,y)&=& h(x,1) \,\exp \left[ \int_1^y  g(x,u) du \right] \nn \,.
\eq
\nind Now suppose $h(\cdot,\cdot)$ does exist and $h(x_0,y_0)\neq 0$ for some $(x_0,y_0)\in\mathbb{R}\times\mathbb{R}_+$,  then the function $B(x,y):=\log |h(x,y)|$ is a twice continuously differentiable function locally defined in a small neighborhood of $(x_0,y_0)$. It is easily seen that in this neighborhood, we have
\be
\partial_yB(x,y)=g(x,y), \,\,\,\partial_xB(x,y)=\frac{1}{2} \left(1+\frac{\sigma'(x)}{\sigma(x)} \right)-\frac{\rho\alpha(y)}{\sigma(x)y}g(x,y).
\label{bEquations}
\ee
However, for $B(x,y)$ to be locally twice continuously differentiable, necessarily, $\partial_x\partial_y B(x,y)=\partial_y\partial_x B(x,y)$, which leads to 
\bq
-\frac{\rho}{1-\rho^2}\frac{y}{2\al(y)}(\sigma'(x)+\sigma''(x))&=&  \partial_y \left[\frac{1}{2}\left(1+\frac{\sigma'(x)}{\sigma(x)}\right)-\frac{\rho\alpha(y)}{\sigma(x)y}g(x,y)\right] \\
&=&\mf{-\frac{\rho}{1-\rho^2}\frac{1}{\sigma(x)}\p_y\left[\frac{\al(y)}{y}\left(\frac{\al'(y)}{2\al(y)}-\frac{1}{2y}-\frac{\mu(y)}{\al(y)^2}
-\frac{\rho y}{2\al(y)}(\sigma(x)+\sigma'(x)) \right)\right]}\nn \\
&=&-\frac{\rho}{1-\rho^2}\frac{1}{\sigma(x)}\p_y\left[\frac{\al(y)}{y}\left(\frac{\al'(y)}{2\al(y)}-\frac{1}{2y}-\frac{\mu(y)}{\al(y)^2}
 \right)\right]\nn \,.
\eq
When $\rho=0$, the above equality holds automatically. If $\rho\neq 0,\pm1$, then the above equality is equivalent to
\bq \left(\sigma'(x)+\sigma''(x)\right)\sigma(x)=\frac{2\al(y)}{y}\p_y\left[\frac{\al(y)}{y} \left(\frac{\al'(y)}{2\al(y)}-\frac{1}{2y}-\frac{\mu(y)}{\al(y)^2}\right)\right]\, ,\eq
which cannot hold unless  both sides are equal to a constant $b$. Suppose this is the case and $b<0$, then $\frac{b}{\ul{\sigma}}<\sigma'(x)+\sigma''(x)<\frac{b}{\bar{\sigma}}<0$.  Multiplying both sides by $e^x$ we obtain
\bq \frac{b}{\ul{\sigma}}e^x \lt e^x(\sigma'(x)+\sigma''(x))=(e^x\sigma'(x))'\le \frac{b}{\bar{\sigma}}e^x \lt 0\, . \eq
If we now integrate both inequalities from $0$ to $x>0$, we obtain that
\bq \frac{b}{\bl{\ul{\sigma}}}(e^x-1) \lt e^x\sigma'(x)-\sigma'(0) \lt \frac{b}{\bar{\sigma}}(e^x-1) \lt 0 \, . \label{eq:C7} \eq
Solving for $\sigma'(x)$ from the above inequalities, we have
\bq e^{-x}\sigma'(0)+\frac{b}{\ul{\sigma}}(1-e^{-x}) \lt \sigma'(x) \lt e^{-x}\sigma'(0)+\frac{b}{\bar{\sigma}}(1-e^{-x})\,. \eq
Letting $x\to\infty$, we see that $\sigma'(x)<\frac{b}{\bar{\sigma}}<0$, which contradicts the assumption that $\sigma$ is smooth and uniformly bounded. Similarly, we can show that $b$ cannot be strictly positive. Hence, the only possibility is when $\sigma'(x)+\sigma''(x)\equiv0$. In this case, the only positive bounded solution is $\sigma(x)\equiv\sigma_0$ for some positive constant $\sigma_0$.

\bs

In conclusion, \ja{a necessary condition} for $h(x,y)$ to exist is that
$\sigma(x)$ is a constant and
\bq
\label{eq:Gatlin}
\frac{\al(y)}{y}(\frac{\al'(y)}{2\al(y)}-\frac{1}{2y}-\frac{\mu(y)}{\al(y)^2})\equiv c
 \eq
 for some constant $c$. \ja{Conversely, if this holds, then, by the criterion for an exact differential (see e.g. page 61 of Donaldson\cite{Don11}), we can find a $B$ satisfying \eqref{bEquations}. Hence if we take $h:=e^B$
it will satisfy \eqref{Aeq1} and \eqref{Aeq2}. }

Re-arranging \eqref{eq:Gatlin}, we obtain
\bq
\mu(y) &=& \frac{\al(y)}{2y} \left[ y\al'(y)-\al(y) \right]-cy\al(y) . \nn \,
\eq
\bl{But substituting the asymptotic behaviour for $\al(y)$ as $y \to \infty$ given in Assumption \ref{Assumption:TechnicalConditions}, we find that $\mu(y)\to \infty$ unless $c\ge0$ ($y\alpha(y)$ is the dominating term). On the other hand, by applying the mean value theorem to $\alpha'(y)$, we have
$\alpha'(y)=A_1+\alpha''(\zeta_y)y$, where $\zeta_y\in(0,y)$ is some point depending on $y$. Hence, as $y\to0$, we have
\[\mu(y)=\frac{A_1}{2}(1+o(1))\alpha''(\zeta_y)y^2-cA_1y^2(1+o(1)).\]
Because $\alpha''(y)\to0$ as $y\to0$, we know that $-cA_1y^2$ is the leading term of $\mu(y)$ as $y\to0$. In order for $\mu(y)\ge0$ for all sufficiently small $y>0$, we must have $c\le 0$. Overall, the only suitable choice of $c$ is zero.
}
\bs
As a consequence, we have that
\bq
\frac{\p_yh}{h}=g(x,y)\equiv g(y):=-\frac{\rho\sigma_0}{2(1-\rho^2)}\frac{y}{\alpha(y)},\quad\frac{\p_xh}{h}=\half-\frac{\rho\alpha(y)}{\sigma_0 y}\frac{\p_yh}{h}=\frac{1}{2(1-\rho^2)}:=C\,,\nn
\eq
\mf{and hence $\p^2 _x h= C^2 h$ and $\p^2_y h= h g' +g \p_yh$.  Using these relations and \eqref{eq:Veexy} we find that}
\bq
V(x,y)&=&\frac{(\half\sigma_0^2y^2\p_x^2+\rho\sigma_0y\alpha(y)\p_x\p_y+\half\alpha(y)^2\p_y^2-\half\sigma_0^2y^2\p_x+\mu(y)\p_y)h}{h}\nn\\
&=&\half\sigma_0^2y^2C(C-1)+\rho\sigma_0y\alpha(y) C g(y)+\half\alpha(y)^2(g(y)^2+g'(y))+\mu(y)g(y)\nn\\
&=:&V(y) \,.
\eq
\bs
Using that $\alpha(y) \sim A_1 y$ as $y \to 0$ and $\alpha(y) \sim B_1 y^p$ as $y \to \infty$ we find that
\bq
V(y) &\sim& \mf{-\frac{\sigma_0^2 y^2}{8(1-\rho^2)}}  \quad \quad \quad \quad (\mathrm{as} \,\, \,y\to \infty \,\,\mathrm{and \,\,as }\,\, y \to 0)\,, \nn \,
\eq
and $V(x,y)<\infty$ for all $x,y$, so $V$ is bounded from above, as required.


\begin{thebibliography}{99}

\sk

\bibitem[AP07]{AP07} Andersen, L.B.G., V.V. Piterbarg, ``Moment explosions in stochastic volatility Models'', \textit{Finance and Stochastics}, 11(1), 29-50, 2007.

\bibitem[AS12]{AS12} Antonov, A. and M. Spector, ``Advanced analytics for the SABR model'', \url{http://ssrn.com/abstract=2026350}, 2012.

\bibitem[Bel81]{Bel81}
 Bellaiche, C., ``Comportement asymptotique de $p_t(\mb{x},\mb{y})$'', \textit{Ast\'{e}rique}, 84-85, 151-187, 1981.

\bibitem[BA88]{BA88} Ben Arous, G., ``Methods de Laplace et de la phase stationnaire sur le space de Wiener'',
\textit{Stochastics}, 25: 125-153, 1988.

\bibitem[BBF02]{BBF}
 Berestycki, H., J. Busca and I. Florent, ``Asymptotics and calibration of local volatility models'', \textit{Quantitative Finance}, 2, 61-69, 2002.

\bibitem[BS13]{BS13} Brunick, G. and S. Shreve, ``Matching an It\^{o} process by a solution of a stochastic differential equation'', \textit{Annals of Applied Probability}, 23, 1584-1628, 2013.

 \bibitem[BBF04]{BBF2}
 Berestycki, H., J. Busca and I. Florent, ``Computing the implied  volatility in stochastic volatility models'' \textit{Communications on  Pure and Applied Mathematics}, 57(10), 1352-1373, 2004.

\bibitem[Chav84]{Chav84} Chavel, I., ``Eigenvalues in Riemannian Geometry'', Pure and Applied Mathematics, vol. 115, Academic Press Inc., Orlando, FL, 1984.

\bibitem[Dav88]{Dav88} Davies, E.B., ``Gaussian upper bounds for the heat kernels of some second-order operators on Riemannian manifolds'', \textit{Journal of Functional Analysis}, 80,
16-32, 1988.

\bibitem[DFJV11]{DFJV11} Deuschel, J.D., P.K. Friz, A. Jacquier, S. Violante, ``Marginal density expansions for diffusions and stochastic volatility, Part II: Theoretical foundations '',  \textit{Communications on  Pure and Applied Mathematics}, 67(2): 321-350, 2014.

\bibitem[DFJV11b]{DFJV11b} Deuschel, J.D., P.K. Friz, A. Jacquier, S. Violante, ``Marginal density expansions for diffusions and stochastic volatility, Part I: Applications'',  \textit{Communications on  Pure and Applied Mathematics}, 67(1): 40-82, 2014.

\bibitem[dW65]{dW65} DeWitt, B.S. ``Dynamical Theory of Groups and Fields'', Gordon and Breach, 1965.

\bibitem[doC92]{doC92} do Carmo, M., ``Riemannian Geometry'', Birkh\"{a}user, 1992.


\bibitem[Don11]{Don11} Donaldson, S., ``Riemann Surfaces'', Oxford University Press, 2011.


\bibitem[Dur04]{Dur04} Durrleman, V., ``From Implied to Spot Volatilities'', PhD thesis, Princeton University, 2004.

\bibitem[FFL12]{FFL12}  Figueroa-L\'{o}pez, J.E. and M. Forde, ``The small-maturity smile for exponential L\'evy models'', \textit{SIAM Journal on Financial Mathematics}, 3, 33-65, 2012.

\bibitem[FGH14]{FGH14}  Figueroa-L\'{o}pez, J.E., R. Gong and C. Houdr\'{e}, ``High-order short-time expansions for ATM option prices of exponential L\'{e}vy models'', \textit{Mathematical Finance}, 26(3), 516-557, 2016.

\bibitem[FGH12]{FGH12} Figueroa-L\'{o}pez, J.E., R. Gong and C. Houdr\'{e}, ``Small-time expansions of the distributions, densities, and option prices of stochastic volatility models with L\'evy jumps'', \textit{Stochastic Processes and their Applications}, 122, 1808-1839, 2012.

\bibitem[FJ11]{FJ11} Forde, M. and A. Jacquier ``Small-time asymptotics for implied volatility under a general Local-Stochastic volatility model'', \textit{Applied Mathematical Finance}, 18, 517-535, 2011.

\bibitem[FJL12]{FJL12} Forde, M., A. Jacquier and R. Lee, ``The small-time smile and term structure of implied volatility under the Heston model'', \textit{SIAM Journal on Financial Mathematics}, 3, 690-708, 2012.

\bibitem[FdM13]{FdM13} Friz, P. and S. DeMarco,  ``Varadhan's formula, conditioned diffusions, and local volatilities'', \url{http://arxiv.org/pdf/1311.1545.pdf}, 2013.

\bibitem[GHLOW12]{GHLOW12} Gatheral, G., E. Hsu, E.P. Laurence, C. Ouyang and T.-H. Wang, ``Asymptotics of implied volatility in local volatility
models'', (2012), \textit{Mathematical Finance}, 22(4), 591-620, 2012.

\bibitem[Gy\"{o}86]{Gyo86} Gy\"{o}ngy, I., ``Mimicking the one-dimensional marginal distributions
of processes having an It\^{o} differential'', \textit{Probability Theory and Related Fields}, 71(4), 501-516, 1986.

\bibitem[HKLW02]{HKLW02} Hagan, P., D. Kumar, A. S. Lesniewski and D.E. Woodward,
 ``Managing smile risk'', \textit{Wilmott Magazine}, 2002.

\bibitem[HL08]{HL08} Henry-Labord\`{e}re, P., ``Analysis, Geometry, and
Modeling in Finance: Advanced Methods in Option Pricing", Chapman  \&  Hall, 2008.

\bibitem[Hsu]{Hsu} Hsu, E.P., ``A brief introduction to Brownian motion on a Riemannian manifold'', unpublished lecture notes.

\bibitem[Hsu02]{Hsu02} Hsu, E.P., ``Stochastic Analysis on Manifolds'', Graduate Studies in Mathematics, vol. 38,
American Mathematical Society, Providence, RI, 2002.

\bibitem[Jost09]{Jost09} Jost, J., ``Riemannian Geometry and Geometric Analysis'', Fifth edition, Springer, 2008.

\bibitem[Jour04]{Jour04} Jourdain, B, ``Loss of martingality in asset price models with lognormal stochastic volatility'',
\textit{preprint Cermics}, 267, 2004.

\bibitem[KS91]{KS91} Karatzas, I. and S. Shreve, ``Brownian Motion and
 Stochastic Calculus'', Springer-Verlag, 1991.

\bibitem[Laur08]{Laur08} Laurence, P., ``Implied volatility, fundamental solutions,
asymptotic analysis and symmetry methods'', Caltech, April 2008.

\bibitem[Laur10]{Laur10} Laurence, P., ``Asymptotics for local volatility and SABR models'', Global Derivatives, Paris, France, 2010.

\bibitem[Lew07]{Lew07} Lewis, A., ``Geometries and smile asymptotics for a class of Stochastic Volatility
models'', \verb"www.optioncity.net", 2007.

\bibitem[LM07]{LM07} Lions, P.-L. and M. Musiela, ``Correlations and bounds for stochastic volatility models'', \textit{Annales de l'Institut Henri Poincare (C) Non Linear Analysis''}, 24(1), 1-16, 2007.

\bibitem[LPP15]{LPP15} Lorig, M., S. Pagliarani and A. Pascucci, ``Explicit implied volatilities for multifactor local-stochastic volatility models'', to appear in \textit{Mathematical Finance}, 2015.

\bibitem[LPP14]{LPP14} Lorig, M., S. Pagliarani and A. Pascucci, ``Analytic expansions for parabolic equations'', \textit{SIAM Journal on Applied Mathematics}, 75(2),  468-491, 2014.

\bibitem[MA73]{MA73} Misner, C.W. and J. Archibald, ``Gravitation'',  W. H. Freeman and Company, 1973.

\bibitem[McK70]{McK70} McKean, H.P. , ``An upper bound to the spectrum of $\Delta$ on a manifold of
negative curvature'', \textit{Journal of Differential Geometry}, 4, 359-366, 1970 .

\bibitem[MO91]{MO91} McAvity, D.M. and H. Osborn, ``A DeWitt expansion of the heat kernel for manifolds with a
boundary'', \textit{Classical Quantum Gravity}, 8, 603-638, 1991.

\bibitem[MP49]{MP49} Minakshisundaram, S. and A. Pleijel, ``Some properties of the eigenfunctions of the Laplace operator
on Riemannian manifolds'', \textit{Canadian Journal of Mathematics}, 1, 242-256, 1949.

\bibitem[Mol75]{Mol75} Molchanov, S., ``Diffusion processes and Riemannian geometry'', \textit{Russian Mathematics Surveys} 30:1, 1-63, 1975.

\bibitem[MY05]{MY05} Matsumoto, H. and Yor, M., ``Exponential functionals of Brownian motion II: Some related diffusion processes'', \textit{Probability Surveys}, 2, 348-384, 2005.


\bibitem[Neel07]{Neel07} Neel, R., ``The small-time asymptotics of the heat kernel at the cut locus'', \textit{Communications in Analysis and  Geometry} 15(4), 845-890, 2007.

\bibitem[NS91]{NS91} Norris, J. and D.W. Stroock , ``Estimates on the fundamental solution to
heat flows with uniformly elliptic coefficients'', \textit{Proceedings  London Mathematical Society}, 62, 375-402, 1991

\bibitem[Olv74]{Olv74} Olver, F.W., ``Asymptotics and Special Functions'', Academic Press, 1974.

\bibitem[Pau10]{Pau10} Paulot, L., ``Asymptotic implied volatility at the second order
with application to the SABR model'', working paper, 2010.

\bibitem[PP14]{PP14} Pagliarani, S. and A. Pascucci, ``Asymptotic expansions for degenerate parabolic Equations'',
\textit{Comptes Rendus Mathematique}, 352(12), 1011-1016, 2014.

\bibitem[SS03]{SS03} Stein, E.M. and R. Sharkarchi, ``Complex Analysis", Princeton University Press, 2003.


\bibitem[Vass03]{Vass03} Vassilevich, D.V., ``Heat kernel expansion: user's manual'', \textit{Physics Reports} 388, 279-360, hep-th/0306138, 2003.

\end{thebibliography}
\end{document}